\documentclass[twoside]{article}
\usepackage{amssymb}
\usepackage{amsmath}
\usepackage[dvipsnames]{xcolor}
\usepackage{qic,epsfig}
\usepackage{algorithm,algorithmic}
\renewenvironment{proof}[1][Proof]{\vspace*{12pt}\noindent{\bf #1: }}

\newtheorem{proposition}{Proposition}
\newtheorem{remark}{Remark}
\newtheorem{problem}{Problem}

\begin{document}
\setlength{\textheight}{8.0truein}    

\runninghead{Infeasibility of constructing a special orthogonal matrix for the deterministic remote preparation... }
            {Wenjie Liu, Zixian Li, and Gonglin Yuan}

\normalsize\textlineskip
\thispagestyle{empty}
\setcounter{page}{1289}


\vspace*{0.88truein}

\alphfootnote

\fpage{1289}

\centerline{\bf
INFEASIBILITY OF CONSTRUCTING A SPECIAL }
\vspace*{0.035truein}
\centerline{\bf ORTHOGONAL MATRIX FOR THE DETERMINISTIC REMOTE }
\vspace*{0.035truein}
\centerline{\bf PREPARATION  STATE OF AN ARBITRARY N-QUBIT STATE}

\vspace*{0.37truein}
\centerline{\footnotesize
WENJIE LIU}
\vspace*{0.015truein}
\centerline{\footnotesize\it Engineering Research Center of Digital Forensics, Ministry of Education, }
\baselineskip=10pt
\centerline{\footnotesize\it Nanjing University of Information Science and Technology,}
\baselineskip=10pt
\centerline{\footnotesize\it Nanjing, 210044, China}

\vspace*{0.015truein}
\centerline{\footnotesize\it School of Computer and Software, Nanjing University of Information Science and Technology,}
\baselineskip=10pt
\centerline{\footnotesize\it Nanjing, 210044, China}
\baselineskip=10pt
\centerline{\footnotesize\it wenjiel@163.com}

\vspace*{10pt}
\centerline{\footnotesize 
ZIXIAN LI}
\vspace*{0.015truein}
\centerline{\footnotesize\it School of Computer and Software, Nanjing University of Information Science and Technology,}
\baselineskip=10pt
\centerline{\footnotesize\it Nanjing, 210044, China}
\baselineskip=10pt
\centerline{\footnotesize\it zixianli157@163.com}

\vspace*{10pt}
\centerline{\footnotesize 
GONGLIN YUAN}
\vspace*{0.015truein}
\centerline{\footnotesize\it School of Mathematics and Information Science, Center for Applied Mathematics of Guangxi,  }
\baselineskip=10pt
\centerline{\footnotesize\it Guangxi University, Nanning, 530004, China}
\baselineskip=10pt
\centerline{\footnotesize\it yuangl0417@126.com}

\vspace*{0.225truein}
\publisher{July 9, 2022}{October 8, 2022}

\vspace*{0.21truein}

\abstracts{
In this paper, we present a polynomial-complexity algorithm to construct a special orthogonal matrix for the deterministic remote state preparation (DRSP) of an arbitrary $n$-qubit state, and prove that if $n > 3$, such matrices do not exist. Firstly, the construction problem is split into two sub-problems, i.e., finding a solution of a semi-orthogonal matrix and generating all semi-orthogonal matrices. Through giving the definitions and properties of the matching operators, it is proved that the orthogonality of a special matrix is equivalent to the cooperation of multiple matching operators, and then the construction problem is reduced to the problem of solving an XOR linear equation system, which reduces the construction complexity from exponential to polynomial level. Having proved that each semi-orthogonal matrix can be simplified into a unique form, we use the proposed algorithm to confirm that the unique form does not have any solution when $n > 3$, which means it is infeasible to construct such a special orthogonal matrix for the DRSP of an arbitrary $n$-qubit state.}{}{}

\vspace*{10pt}


\vspace*{1pt}\textlineskip    
\section{Introduction}        
\noindent
Remote state preparation (RSP)\cite{2000LoCla, 2001BennettRem, 2003BerryOpt, 2005BennettRem} is a method for quantum state transmission via quantum entanglement in which several classical channels are required. In RSP, the sender (known as Alice) knows the state to be sent but does not need to prepare it. RSP is an important issue in the field of quantum information because of the lower cost of classical communication, and various RSP schemes of a quantum state with a certain number of qubits, such as a single-qubit state\cite{2010BarreiroRem, 2015WangCon, 2021ShiCon,2021WangRem, 2021PengBid}, two-qubit state\cite{2017AdepojuJoi,2021SunDou,2021QianEff}, three-qubit state\cite{2011XiaoJoi,2012ZhanDet, 2013ZhanDet,2016RaRem,2017ChenCon}, four-qubit state\cite{2006DaiCla,2019XueRem,2020SangOpt,2020DuDet}, and five-qubit state\cite{2016ChenEco}, have been proposed. In addition, since the theoretical probability of success is $100\%$, deterministic RSP (DRSP) schemes\cite{2021ShiCon,2021SunDou,2012ZhanDet,2013ZhanDet,2020SangOpt,2020DuDet,2016ChenEco} are more concerned than probabilistic ones\cite{2015WangCon,2011XiaoJoi,2017ChenCon,2019XueRem}.

For the sake of generality, several DRSP schemes of an arbitrary $n$-qubit state (here, ``arbitrary" means any number of qubits) are proposed\cite{2018WeiOpt, 2019JiangCon, 2019ZhouDet, 2019WeiDet, 2021PengPer, 2022ZhaoMul}. There is a similar and key element in all of them: a kind of special $2^n\times 2^n$ real-parameter orthogonal matrices, where each column is consisted of permutations of real coefficients $ \pm a_0,\pm a_1,\cdots,\pm a_{2^n-1}$ of the prepared state. These matrices have been used in several schemes\cite{2021SunDou,2012ZhanDet, 2013ZhanDet, 2019XueRem} for a certain number of qubits, and all of them are constructed through accidental patchwork or violent traversal. For an arbitrary $n$, it is necessary to design a general method to construct such a special orthogonal matrix. In 2018, Wei et al.\cite{2018WeiOpt} proposed an algorithm to construct such a matrix for an arbitrary $n$ that limits the traversal range, but it is essentially still an exponential violent traversal. In this paper, we present an polynomial-complexity general algorithm to construct such a matrix for the DRSP of an arbitrary $n$-qubit state, where the construction problem is reduced to the problem of solving an XOR linear equation system. In addition, we prove that such special orthogonal matrices do not exist for $4$ or more qubits through proving that each semi-orthogonal matrix (formally defined in the text) can be simplified into a unique form and then using the proposed algorithm to confirm that the unique form does not have any solution. As a result, we prove that it is infeasible to construct such a orthogonal matrix for the DRSP of an arbitrary $n$-qubit state.

The rest of this paper is organized as follows: in Section~\ref{sec2}, these schemes are briefly reviewed and analyzed. The problem of constructing a special orthogonal matrix is defined and split in Section~\ref{sec3}. The construction algorithm is presented in Section~\ref{sec4}. The infeasibility proof is presented in Section~\ref{sec5}. We conclude the paper in Section~\ref{sec6}.

\section{Schemes review and analysis}\label{sec2}
\noindent
In this section, we briefly review the DRSP of a real-parameter state via maximally two-qubit states as a representative of the above schemes and focus on the review and analysis of Wei et al.s' algorithm\cite{2018WeiOpt}.

\subsection{DRSP of a real-parameter state via maximally two-qubit states}\label{subsec2.1}
\noindent
Assume a sender Alice and a receiver Bob both have a set of $n$ particles respectively, where the particles $j=1,3,5,\cdots, 2n-1$ belong to Alice and $k=2,4,6,\cdots, 2n$ to Bob. There is a quantum entanglement between each particle pair $(j,k)$ held by Alice and Bob, respectively. Each of these pairs forms a maximally two-qubit entangled state which can be 

\noindent
\begin{equation}
\left | \Psi \right \rangle_{jk} =\frac{1}{\sqrt{2} } \left ( \left | 00  \right \rangle +\left | 11  \right \rangle  \right )_{jk},j=1,3,\cdots,2n-1, k=2,4,\cdots,2n\,.
\end{equation}
or other similar Bell states. Thus, we have a total state as $\left | \Psi \right \rangle _{total} =\left | \Psi \right \rangle _{12}\otimes \left | \Psi \right \rangle _{34} \otimes \cdots \otimes \left | \Psi \right \rangle _{(2n-1)2n}$. Now Alice wants to transmit a state $\left | \psi  \right \rangle =\sum_{i=0}^{2^n-1} a_i e^{\imath \varphi _i } \left | i  \right \rangle$ called an arbitrary $n$-qubit \textbf{general} state, where $\imath$ is the imaginary unit, $a_i$, $\varphi_i$ are all arbitrary real numbers and $\sum_{i=0}^{2^n-1} a_i^2=1 $. The follows are the general processes of the RSP via maximally two-qubit states.

\textit{Step 1} Alice needs to find a special projective measurement basis $\{ \left |\tau _i\right\rangle | i=0,1,\cdots,2^n-1 \}$ in the Hilbert space of the particles $j$. Therefore, the total state can be written as 

\noindent
\begin{equation}
\left | \Psi  \right \rangle _{total}=\frac{1}{\sqrt{2^n}}\sum_{i=0}^{2^n-1}\left | \tau_i  \right \rangle _{13\cdots (2n-1)} \otimes\left | \theta _i  \right \rangle_{24\cdots 2n},
\end{equation}
where 
\noindent
\begin{equation}
 \left | \theta _i  \right \rangle_{B}= \left | \tau _i  \right \rangle_A^{\dagger }\otimes\left | \Psi  \right \rangle_{AB},  i=0,1,\cdots,2^n-1
\end{equation}
is a basis in Hilbert space of particles $k$. Now Alice needs to measure her particles on the measurement basis $\left \{ \left |\tau _i\right\rangle \right \}$.

\textit{Step 2} After the measurement, Bob’s particles will collapse into $\left \{ \left |\theta _i\right\rangle \right \}$ where $i$ is one of the $16$ results. Meanwhile, Alice informs the measurement result $i$ to Bob via classical channels.

\textit{Step 3} According to the value of $i$, Bob needs to select a previously agreed unitary operator $U_i$ to act on $\left |\theta _i\right\rangle$, to restore the target state $ \left |\psi\right\rangle$, i.e., $U_i \left |\theta _i\right\rangle=\left |\psi\right\rangle$.

In the last step, we should recover the target state $ \left |\psi\right\rangle$ by using the unitary operator $U_{i}$ independent of $ \left |\psi\right\rangle$. For an arbitrary $n$-qubit state it is hard to realize with a high success probability, e.g., if letting $\left |\tau _0\right\rangle=\left |\psi\right\rangle$, then only in the case of $i=0$, it can be recovered by using Pauli operators with a probability $\frac{1}{16}$\cite{2019XueRem}. 

If only considering a \textbf{real-parameter} state, i.e., $\left | \psi  \right \rangle =\sum_{i=0}^{2^n-1} a_i \left| i  \right \rangle$, it is proposed\cite{2018WeiOpt} to use a complete orthogonal basis 
$\left | \tau_i  \right \rangle =U[\Theta_n^n] 
\begin{pmatrix} \left | 0 \right \rangle,\left | 1 \right \rangle,\cdots,\left | 2^n-1 \right \rangle \end{pmatrix}^T$, where the elements of the $2^n\times 2^n$ orthogonal matrix $U[\Theta_n^n]$ are consisted of permutations of $ \pm a_0,\pm a_1,\cdots,\pm a_{2^n-1}$. Thus, the $ \left |\psi\right\rangle$ in the above scheme can be recovered by simply element-rearranging and phase-inversion. If any above orthogonal matrix $U[\Theta_n^n]$ is constructed, then the RSP of an arbitrary $n$-qubit real-parameter state can be realized simply with a probability of $100\% $.

Although the above scheme is only applicable to the DRSP of a real-parameter state, there are also several DRSP schemes\cite{2019JiangCon,2019WeiDet,2021PengPer, 2022ZhaoMul} of a general state that use the above orthogonal matrix. Therefore, the construction of such an orthogonal matrix is a key point to realize the DRSP of an arbitrary $n$-qubit state.

\subsection{Wei et al.'s algorithm to construct a special orthogonal matrix}\label{subsec2.2}
\noindent
Wei et al.\cite{2018WeiOpt} proposed an algorithm to construct an above orthogonal $2^n\times 2^n$ matrix $U[\Theta^n_n]$ as follows.

\textit{Step 1} Set the elements of the $1$st and $(2^{n-1}+1)$th rows together with the $1$st and $(2^{n-1}+1)$th columns

\noindent
\begin{equation}\label{eq1}
U[\Theta_n^n]=\begin{pmatrix}
 a_0 & a_1 & \cdots & a_{2^{n-1}-1} & a_{2^{n-1}} & a_{2^{n-1}+1} & \cdots & a_{2^n-1}\\
 a_1 & {\color{Red} \ddots}   & {\color{Red} \cdots} & {\color{Red} \cdots} & a_{2^{n-1}+1} & {\color{Blue} \cdots}  & {\color{Blue} \cdots} & {\color{Blue} \cdots}\\
 \vdots & {\color{Red} \vdots} & {\color{Red} \ddots} & {\color{Red} \cdots} & \cdots & {\color{Blue} \cdots} & {\color{Blue} \cdots} & {\color{Blue} \cdots}\\
 a_{2^{n-1}-1} & {\color{Red} \vdots} & {\color{Red} \vdots} & {\color{Red} \ddots} & a_{2^n-1} & {\color{Blue} \cdots} & {\color{Blue} \cdots} & {\color{Blue} \cdots}\\
 a_{2^{n-1}} & -a_{2^{n-1}+1} & \vdots & -a_{2^n-1} & -a_0 & a_1 & \cdots & a_{2^{n-1}-1}\\
 a_{2^{n-1}+1} & {\color{Green} \vdots} & {\color{Green} \vdots} & {\color{Green} \vdots} & -a_1 & {\color{Blue} \ddots} & {\color{Blue} \cdots} & {\color{Blue} \cdots}\\
 \vdots & {\color{Green} \vdots} & {\color{Green} \vdots} & {\color{Green} \vdots} & \vdots & {\color{Blue} \vdots} & {\color{Blue} \ddots} & {\color{Blue} \cdots}\\
 a_{2^n-1} & {\color{Green} \vdots} & {\color{Green} \vdots} & {\color{Green} \vdots} & -a_{2^{n-1}-1} & {\color{Blue} \vdots} & {\color{Blue} \vdots} &{\color{Blue} \ddots}
\end{pmatrix}.
\end{equation}

When $n=1$, we can obtain the matrix $U[\Theta^1_1]$ from this step, and the following processes would be ignored.

\textit{Step 2} If $n \le 2$, then set that $U[\Theta^n_n](i,j)=U[\Theta^{n-1}_{n-1}](i,j)$, here $2\le i, j \le2^{n-1}$. The parameter $U[\Theta^n_n](i,j)$ presents the element in the $i$th row and $j$th column of the matrix $U[\Theta^n_n]$. After that, the red items on the above equation can be determined.

\textit{Step 3} According to this equation $\sum_{i=1}^{2^n}U[\Theta^n_n](i,j)\cdot U[\Theta^n_n](i, 2^{n-1}+1) = 0 (j\le 2^{n-1})$, we can obtain the elements from the $1$st column to $2^{n-1}$th column, i.e., the green elements of $U[\Theta^n_n](i,j)$ could be calculated.

\textit{Step 4} Based on $\sum_{i=1}^{2^n}U[\Theta^n_n](i,j)\cdot U[\Theta^n_n](i, k) = 0 (k=1, 2^{n-1}+1, 2\le j\le 2^n)$ the elements from the $1$st column to $2^n$th column could be completed; namely, the blue elements of the matrix $U[\Theta^n_n](i,j)$ can be fulfilled.

\subsection{Analysis of Wei et al.'s algorithm}\label{subsec2.3}
\noindent
As can be seen, the last two steps of the above algorithm do not give clear execution steps, but only the constraints of the solution, which is essentially a violent traversal. Let $N=2^n$. Even if the element arrangement is determined, it needs to traverse $2^{\left(\frac{N}{2}-1 \right)^2}+2^{2\left(\frac{N}{2}-1\right)^2}=O (2^{\frac{N^2}{2}})$ times to determine the distribution of minus signs in the worst case, because those signs can be regarded as Boolean variables. As an instance, let $n=4$, i.e., $N=16$, then we need to traverse $2^{\left(\frac{16}{2}-1 \right)^2}+2^{2\left(\frac{16}{2}-1\right)^2}\approx 3.169\times 10^{29}$ times at worst, which is impossible. Therefore, for the RSP with a large $n$ ($n>3$), the algorithm needs unimaginable time.

In fact, we will prove that for $n>3$ the above orthogonal matrix does not exist, which means that the above algorithm will not output any solution. To prove it, we will first propose a more efficient construction algorithm with polynomial time and space complexity for $N$.

\section{Problem definition and split of constructing a special orthogonal matrix}\label{sec3}
\noindent
In this section, we will define the problem of constructing an above orthogonal matrix as Problem~\ref{pro1} in Section~\ref{subsec3.1} and split it into two sub-problems as Problem~\ref{pro2} and Problem~\ref{pro3} in Section~\ref{subsec3.2}.

\subsection{Problem definition}\label{subsec3.1}
\noindent
We define the construction problem as follows.

\begin{definition}[\textbf{Special orthogonal matrix}]
Given a set $U=\left \{ a_0,a_1,\cdots,a_{15} \right\}$ of real parameters of size $N$, where $N$ is a positive integer power of 2 and $\sum_{i=0}^{N-1}a_i^2=1$. A \textbf{special orthogonal matrix} $S_N$ is an orthogonal matrix, where each column vector $\textbf{\textit{s}}^{(i)}$ is obtained by adding positive and negative signs to the elements of $U$ and arranging them arbitrarily without repeating them, i.e., 
$\textbf{\textit{s}}^{(i)}=\begin{pmatrix} \left( -1 \right)^{\sigma_{0i}}a_{\alpha_{0i}} ,
\left( -1 \right)^{\sigma_{1i}}a_{\alpha_{1i}} , \cdots, \left( -1 \right)^{\sigma_{(N-1)i}}a_{\alpha_{(N-1)i}} \end{pmatrix}^T$
where $\left \{ \alpha_{0i},\alpha_{1i},\cdots,\alpha_{(N-1)i} \right \}$ is a permutation of index sequence $\left [N \right ]=\left \{ 0,1,\cdots,N-1 \right \}$, and $\sigma_{0i},\sigma_{1i},\cdots,\sigma_{(N-1)i}$ are each valued as $0,1$, e.g., $\textbf{\textit{s}}=\begin{pmatrix} -a_1, a_0,-a_3,-a_2 ,\cdots,-a_{N-2} \end{pmatrix}^T$.
\end{definition}

\begin{problem}\label{pro1}(\textbf{Construct a special orthogonal matrix})
Given a set of real parameters $U$ of size $N$ where $N$ is a positive integer power of $2$ and $\sum_{i=0}^{N-1}a_i^2=1$, construct an $N$-order special orthogonal matrix.
\end{problem}

\subsection{Problem split}\label{subsec3.2}
\noindent
We have a definition as follows.

\begin{definition}[\textbf{Semi-orthogonal matrix}]
An $N$-order matrix $\left| S_N\right|$ is a \textbf{semi-orthogonal matrix} if:
\begin{itemlist}
\item Each column vector $\textbf{\textit{s}}^{(i)}=\begin{pmatrix} a_{\alpha_{0i}} ,a_{\alpha_{1i}} ,\cdots,a_{\alpha_{(N-1)i}} \end{pmatrix}^T$, where $\left \{ \alpha_{0i},\alpha_{1i},\cdots,\alpha_{(N-1)i} \right \}$ is a permutation of index sequence $\left [N \right ]$.

\item For any two column vectors $\textbf{\textit{s}}^{(i)},\textbf{\textit{s}}^{(j)}$, all elements in $\textbf{\textit{s}}^{(i)}$ can be divided into several $2$-tuples, and in $\textbf{\textit{s}}^{(j)}$, the elements in each $2$-tuple are and only exchanged.
\end{itemlist}
In addition, if for a subset $V$ of $U$, a matrix $\left| S_V \right|$ of order $\left | V \right |$ satisfies the above description, then it is also called semi-orthogonal.
\end{definition}

The following basic proposition can be given.

\begin{proposition}\label{prop3.1}
Given two column vectors

\noindent
\begin{equation}
\begin{aligned}
\textbf{\textit{s}}^{(i)}=\begin{pmatrix}\left( -1 \right) ^{\sigma_{0i}}a_{\alpha_{0i}}\\
\left( -1 \right)^{\sigma_{1i}}a_{\alpha_{1i}}\\ \vdots\\
\left( -1 \right)^{\sigma_{(N-1)i}}a_{\alpha_{(N-1)i}} \end{pmatrix},
\textbf{\textit{s}}^{(j)}=\begin{pmatrix}\left( -1 \right)^{\sigma_{0j}}a_{\alpha_{0j}} \\
\left( -1 \right)^{\sigma_{1j}}a_{\alpha_{1j}}\\ 
\vdots\\
\left( -1 \right)^{\sigma_{(N-1)j}}a_{\alpha_{(N-1)j}} \end{pmatrix}
\end{aligned},
\end{equation}
then the necessary and sufficient condition for $\textbf{\textit{s}}^{(i)}\perp  \textbf{\textit{s}}^{(j)}$ is that all elements of $\textbf{\textit{s}}^{(i)}$ can be divided into several $2$-tuples, and in $\textbf{\textit{s}}^{(j)}$, the elements in each 2-tuple are exchanged and one of them is multiplied by $-1$. i.e.,
\begin{itemlist}
\item $\forall k\in \left[ N\right ]$, $ \alpha_{ki} \ne\alpha_{kj}$;

\item $\forall k,l\in \left [N \right ]$ and $k\ne l$, if $\alpha_{ki} =\alpha_{lj}$, then $\alpha_{li} =\alpha_{kj}$ and $\sigma_{ki}\oplus\sigma_{kj}= \sigma_{li}\oplus\sigma_{lj}\oplus 1$, where ``$\oplus$" means XOR.
\end{itemlist}
\end{proposition}

\begin{proof}
\begin{itemlist}
\item (\textit{Sufficiency}) $\forall k\in\left [N\right]$, it must be able to find $l\in\left [N \right ]$ and $l\ne k$ so that $\alpha_{ki} =\alpha_{lj}$, because $\left\{\alpha _{ki} | k\in\left [N \right ] \right \}$ and $\left\{\alpha _{lj} | l\in \left [N \right ]\right \}$ are both permutations of $\left [N \right ]$. \\
Without loss of generality, if $\alpha_{ki} =\alpha_{lj}$, then $\alpha_{li} =\alpha_{kj}$ and $\sigma_{ki}\oplus\sigma_{kj}= \sigma_{li}\oplus\sigma_{lj}\oplus 1$, thus in the inner product expression of $\textbf{\textit{s}}^{(i)}, \textbf{\textit{s}}^{(j)}$, there must be a term 

\noindent
\begin{equation}
\begin{aligned}
 &\left( -1 \right) ^{\sigma_{ki}}a_{\alpha_{ki}}\left( -1 \right) ^{\sigma_{kj}}a_{\alpha_{kj}}+\left( -1 \right) ^{\sigma_{li}}a_{\alpha_{li}}\left( -1 \right) ^{\sigma_{lj}}a_{\alpha_{lj}} \\
&=\left( -1 \right) ^{\sigma_{ki}\oplus \sigma_{kj}}a_{\alpha_{ki}}a_{\alpha_{kj}}+\left( -1 \right) ^{\sigma_{li}\oplus \sigma_{lj}}a_{\alpha_{li}}a_{\alpha_{lj}}\\
&=\left( -1 \right) ^{\sigma_{li}\oplus \sigma_{lj}\oplus 1}a_{\alpha_{li}}a_{\alpha_{lj}}+\left( -1 \right) ^{\sigma_{li}\oplus \sigma_{lj}}a_{\alpha_{li}}a_{\alpha_{lj}}\\
&=\left( -1 \right) ^{\sigma_{li}\oplus \sigma_{lj}}a_{\alpha_{li}}a_{\alpha_{lj}}\left( -1+1 \right)\\
&=0.
\end{aligned}
\end{equation}

It is consistent for each $2$-tuple in $\textbf{\textit{s}}^{(i)}$, so the inner product $\textbf{\textit{s}}^{(i)\dagger}\textbf{\textit{s}}^{(j)}=0$, i.e., $\textbf{\textit{s}}^{(i)}\perp \textbf{\textit{s}}^{(j)}$.

\item (\textit{Necessity}) Without losing generality, suppose $\alpha_{ki}=\alpha_{kj}$, then there is a term $\pm \alpha_{ki}\alpha_{kj}$ in the inner product expression without any corresponding term $\mp \alpha_{ki}\alpha_{kj}$ because elements in a permutation cannot be repeated. Thus, for any element values, the inner product of $\textbf{\textit{s}}^{(i)},\textbf{\textit{s}}^{(j)}$ cannot be guaranteed to be 0, so $\forall k\in \left [N \right ]$, $ \alpha_{ki} \ne\alpha_{kj}$.

Assume $\alpha_{ki}=\alpha_{lj}$ where $k\ne l$. There must be a term $\left( -1 \right) ^{\sigma_{ki}\oplus \sigma_{kj}}a_{\alpha_{ki}}a_{\alpha_{kj}}+\left( -1 \right) ^{\sigma_{li}\oplus \sigma_{lj}}a_{\alpha_{li}}a_{\alpha_{lj}}$
in the inner product expression $\textbf{\textit{s}}^{(i)\dagger} \textbf{\textit{s}}^{(j)}$. Since the elements are valued arbitrarily, $\left( -1 \right) ^{\sigma_{ki}\oplus \sigma_{kj}}a_{\alpha_{ki}}a_{\alpha_{kj}}$ must correspond to $\left( -1 \right) ^{\sigma_{ki}\oplus \sigma_{kj}\oplus 1}a_{\alpha_{ki}}a_{\alpha_{kj}}$ to be offset into $0$. Because $\alpha_{ki}=\alpha_{lj}$, 
if $\alpha_{li} =\alpha_{kj}$ and $\sigma_{ki}\oplus\sigma_{kj}= \sigma_{li}\oplus\sigma_{lj}\oplus 1$, then $\left( -1 \right) ^{\sigma_{li}\oplus \sigma_{lj}}a_{\alpha_{li}}a_{\alpha_{lj}}=\left( -1 \right) ^{\sigma_{ki}\oplus \sigma_{kj}\oplus 1}a_{\alpha_{ki}}a_{\alpha_{kj}}$, which can be offset by $\left( -1 \right) ^{\sigma_{ki}\oplus \sigma_{kj}}a_{\alpha_{ki}}a_{\alpha_{kj}}$; if $\alpha_{li} \ne\alpha_{kj}$, it’s impossible to find another term $a_{\alpha_{mi}}a_{\alpha_{mj}}=a_{\alpha_{ki}}a_{\alpha_{kj}}$ because $\forall m \in \left [N \right ]$, $a_{\alpha_{mj}}\ne a_{\alpha_{lj}} = a_{\alpha_{ki}}$, and $a_{\alpha_{mj}}\ne a_{\alpha_{kj}} $. Thus, $\alpha_{li} =\alpha_{kj}$ and $\sigma_{ki}\oplus\sigma_{kj}= \sigma_{li}\oplus\sigma_{lj}\oplus 1$ hold $\Box$.
\end{itemlist}
\end{proof}

\begin{remark}\label{rem3.1}
\begin{romanlist}
\item According to Proposition~\ref{prop3.1}, the process of obtaining another orthogonal column vector $b$ from one column vector $a$ is as follows: divide the elements in $a$ into several $2$-tuples, and for each $2$-tuple, exchange the elements and multiply one of them by $-1$.

\item By (i), if removing the signs of all the elements of a special orthogonal matrix $S$, it obviously becomes a semi-orthogonal matrix, denoted as $\left | S \right |$. However, a semi-orthogonal matrix may correspond to one or more special orthogonal matrices, or even not to any special orthogonal matrix.

\item Most schemes using special orthogonal matrices choose to piece up them by using smaller blocks, e.g., Wei et al.'s scheme\cite{2018WeiOpt}. For example, the $2^n\times 2^n$ matrix in Eq.~(\ref{eq1}) can be pieced up by $4$ smaller $2^{\frac{n}{2}}\times 2^{\frac{n}{2}}$ orthogonal matrices. However, it should be considered that the elements arrangement of a special orthogonal matrix is out of order, and the feasibility of constructing without using blocks has not been logically denied. We will prove that the method based on blocks is universal because all special orthogonal matrices can be simplified into forms pieced up by blocks.
\end{romanlist}
\end{remark}

As shown in Remark~\ref{rem3.1}, each special orthogonal matrix corresponds to a semi-orthogonal matrix. If we have a general method to determine whether a semi-orthogonal matrix has a \textbf{solution} (i.e., the corresponding special orthogonal matrix), and generate all semi-orthogonal matrices, then we can determine whether each semi-orthogonal matrix has a solution, so as to find at least one special orthogonal matrix, or deny its existence. In fact, for the case of $n > 3$, there is indeed no solution for each $N$-order semi-orthogonal matrix, which is proved in Section~\ref{sec5}. Therefore, Problem~\ref{pro1} can be split into the following two sub-problems. 

\begin{problem}\label{pro2}(\textbf{Find a solution of a semi-orthogonal matrix}) 
Given any $N$-order semi-orthogonal matrix, find one of the corresponding special orthogonal matrices.
\end{problem}

\begin{problem}\label{pro3}(\textbf{Generate all semi-orthogonal matrices})
Given a positive integer $n$, generate all $2^n$-order semi-orthogonal matrices.
\end{problem}

Problem~\ref{pro2} will be solved in Section~\ref{sec4}, while Problem~\ref{pro3} will be solved in Section~\ref{sec5}.

\section{Proposed algorithm to construct a special orthogonal matrix}\label{sec4}
\noindent
In this section, we present the algorithm for Problem~\ref{pro2}. This section is arranged as follows: in Section~\ref{subsec4.1}, we define the matching operator and show several property of it. It is proved that the orthogonality of every two vectors in a special orthogonal matrix is equivalent to the cooperation of their matching operators, and thus to the the establishment of several matching equations in Section~\ref{subsec4.2}, and then Problem~\ref{pro2} is reduced to the problem of solving an XOR linear equation system in Section~\ref{subsec4.3}. The construction algorithm is presented also in Section~\ref{subsec4.3}.

\subsection{Matching operator}\label{subsec4.1}
\noindent
We have the following definitions.

\begin{definition} [\textbf{Matching operation}]
Since the elements in a column vector are one-to-one corresponding to the row indexes, the process of obtaining another orthogonal column vector $b$ from one column vector $a$ is equivalent to dividing $N$ row indexes in $a$ into $\frac{N}{2}$ row index $2$-tuples, and for each $2$-tuple, exchanging the two rows and multiplying one row of them by $-1$. We call the above operation a \textbf{matching} operation.
\end{definition}

\begin{definition} [\textbf{Couple}]
A \textbf{couple} $\left \langle i,j \right \rangle $ is defined as an operation that exchanges two rows $i,j$ of a column vector and then  multiplies one of them by $-1$. A Boolean variable for each couple is defined: $\left \langle i,j \right \rangle =0,1$, where $0$ means multiplying the row $i$ (i.e., the previous $j$) by $-1$ after the exchanging, and $1$ means the opposite.
\end{definition}

\begin{definition} [\textbf{Division}]
A matching operation $M$ is always a combination of $\frac{N}{2}$ couples that do not overlap, and the set $D(M)$ of these couples is defined as the \textbf{division} of $M$.
\end{definition}

\begin{definition}[\textbf{Scattered matrix}]
A square matrix is called a \textbf{scattered matrix}, if in each row or column there is only one $\pm1$, and the other elements are $0$. These $\pm1$ are called \textbf{scattered points}.
\end{definition}

\begin{remark}\label{rem4.1}
\begin{romanlist}
\item The value of a couple $\left \langle i,j \right \rangle $ means a ``pointing", i.e., $0$ points to $i$ and $1$ points to $j$. We denote multiplying row $i$ by $-1$ as $-i$ (e.g., $i\mapsto j$ means mapping $i$ to $j$ and then multiplying the new $j$ by $-1$), then there are two mappings: $i\overset{\left\langle i,j\right\rangle}{\mapsto} {(-1)}^{\left\langle i,j\right\rangle}j,j\overset{\left\langle j,i\right\rangle}{\mapsto} {(-1)}^{\left\langle j,i\right\rangle}i$, where $\left\langle i,j\right\rangle$ and $\left\langle j,i\right\rangle$ describe the same operation, denoted as  $\left\langle i,j\right\rangle\equiv \left\langle j,i\right\rangle$. Obviously,  $\left\langle i,j\right\rangle=\left\langle j,i\right\rangle\oplus 1$.

\item $\overline{\left\langle i,j\right\rangle}$ is defined as a new couple meaning exchanging rows $i,j$ as well but multiplying the row contrary to $\left\langle i,j\right\rangle$ by $-1$, i.e., $\overline{\left\langle i,j\right\rangle}=\left\langle i,j\right\rangle\oplus 1$. Note no thinking that $\overline{\left\langle i,j\right\rangle}\equiv\left\langle j,i\right\rangle$. 

\item One couple $\left \langle i,j \right \rangle $ always corresponds to one bijection $\sigma : \left \{ i\right \} \leftrightarrow  \left \{ j\right \} $, but one bijection $\sigma$ corresponds to two couples: $\left \langle i,j \right \rangle $ and $\overline{\left\langle i,j\right\rangle}$. If the operation of multiplying $-1$ is not considered, the couple $\left \langle i,j \right \rangle $ can be considered as the bijection $\sigma$.

\item A couple $\left \langle i,j \right \rangle $ one-to-one corresponds to an $N$-order matrix $E_{\left \langle i,j\right\rangle}(-1)E(i,j)$, where $E(i,j)$ is a row-exchanging matrix and $E_{\left \langle i,j\right\rangle}(-1)$ is a row-multiplying matrix that multiplies the row pointed by $\left \langle i,j \right \rangle $ ($0$ points to $i$, $1$ points to $j$) by $-1$. Note that $E(i,j)$ and $E_{\left \langle i,j\right\rangle}(-1)$ are both scattered matrices.
\end{romanlist}
\end{remark}

A matching operation $M$ one-to-one corresponds to an $N$-order matrix, i.e., the product of some elementary matrices as
\noindent
\begin{equation}\label{eq4}
M=\prod_{\left\langle i,j \right \rangle \in D(M)}\left [E_{\left \langle i,j\right\rangle}(-1)E(i,j)\right]=\prod_{\left\langle i,j \right \rangle \in D(M)}E_{\left \langle i,j\right\rangle}(-1)\prod_{\left\langle i,j \right \rangle \in D(M)}E(i,j),
\end{equation}
which called the \textbf{matching operator} corresponding to the matching operation. By Proposition~\ref{prop3.1}, if a matching operator $M$ can be found such that the column vector $b=Ma$, then $b\perp a$; conversely, if $b\perp a$, a corresponding matching operator $M$ can be found.

We have the following proposition about some properties of matching operators, which will be widely used.

\begin{proposition}\label{prop4.1}
The following properties of a matching operator $M$ hold:
\begin{romanlist}
\item (\textit{Anti-self-reversibility}) $M^{-1}=-M$;
\item (\textit{Anti-symmetry}) $M^T=-M$, and all elements on the main diagonal of $M$ are $0$;
\item $-M$ is also a matching operator;
\item A scattered matrix must be a matching operator if it’s anti-self-reversible or anti-symmetric.
\end{romanlist}
\end{proposition}

To prove the above proposition, we first need the following lemmas.

\begin{lemma}\label{lemma4.1}
The following equations hold:
\begin{romanlist}
\item $E(i,j)^2=E_{\left \langle i,j\right\rangle}(-1)^2=I$,
\item $E(i,j)E_{\left \langle i,j\right\rangle}(-1)=E_{\overline{\left \langle i,j\right\rangle}}(-1)E(i,j)$,
\item $\prod_{\left\langle i,j \right \rangle \in D(M)}E_{\left \langle i,j\right\rangle}(-1)\prod_{\left\langle i,j \right \rangle \in D(M)}E_{\overline{\left \langle i,j\right\rangle}}(-1)=-I$.
\end{romanlist}
where $I$ is the identity matrix and $M$ is a matching operator.
\end{lemma}

\begin{proof}
\begin{romanlist}
\item Obviously according to the definition of the elementary matrix.
\item Assume $\left \langle i,j \right \rangle=0$, then $E(i,j)E_{\left \langle i,j\right\rangle}(-1)$ means multiplying by $-1$ to $i$ and then exchanging $i$ to $j$, and $E_{\overline{\left \langle i,j\right\rangle}}(-1)E(i,j)$ means exchanging $i$ to $j$ and then multiplying by $-1$ to the new $j$, which are the same. Similarly in the case of $\left \langle i,j \right \rangle=1$.
\item For each couple $\left \langle i,j \right \rangle$, $E_{\left \langle i,j\right\rangle}(-1)E_{\overline{\left \langle i,j\right\rangle}}(-1)$ means multiplying by $-1$ for the two rows $i,j$. Consider all couples in $D(M)$, then each row is multiplied by $-1$, i.e., the operation is equivalent to $-I$ $\Box$.
\end{romanlist}
\end{proof}

\begin{lemma}\label{lemma4.2}
The following properties of a scattered matrix $A$ hold:
\begin{romanlist}
\item If $B$ is also a scattered matrix, then $AB$ is still scattered.
\item $A$ is orthogonal.
\end{romanlist}
\end{lemma}

\begin{proof}
\begin{romanlist}
\item $\forall i$, there is a scattered point $(i,j)$ in $A$ and $(j,k)$ in $B$, thus $(i,k)$ in $AB$. Similarly, $\forall k$, there is a scattered point $(i,k)$ in $B$. $\forall i$, there is not any other scattered point $(i,j')$ in $A$, not any $(j,k')$ in $B$, thus not any $(i,k')$ in $AB$. Similarly, there is not any other scattered point $(i',k)$ in $AB$. Therefore, each row or column in $AB$ has only one scattered point, so $AB$ is a scattered matrix.
\item The inner product between every two column vectors of the scattered matrix $A$ is $0$, and each column vector is a unit vector, so $A$ is orthogonal, i.e., $A^{-1}=A^T$ $\Box$.
\end{romanlist}
\end{proof}

By Lemma~\ref{lemma4.2} and Eq.~(\ref{eq4}), any matching operator $M$ is a scattered matrix as well. Therefore, we can prove Proposition~\ref{prop4.1}.

\begin{proof}[Proof of Proposition~{\upshape\ref{prop4.1}}]
\begin{romanlist}
\item Let matching $M=\prod_{\left\langle i,j \right \rangle \in D(M)}E_{\left \langle i,j\right\rangle}(-1)\prod_{\left\langle i,j \right \rangle \in D(M)}E(i,j)$. By Lemma~\ref{lemma4.1}, then 
\noindent
\begin{equation}
\begin{aligned}
 &MM=\left [ \prod_{\left\langle i,j \right \rangle \in D(M)}E_{\left \langle i,j\right\rangle}(-1)\prod_{\left\langle i,j \right \rangle \in D(M)}E(i,j)\right ]\left [ \prod_{\left\langle i,j \right \rangle \in D(M)}E_{\left \langle i,j\right\rangle}(-1)\prod_{\left\langle i,j \right \rangle \in D(M)}E(i,j)\right ] \\
&=\left [ \prod_{\left\langle i,j \right \rangle \in D(M)}E_{\left \langle i,j\right\rangle}(-1)\prod_{\left\langle i,j \right \rangle \in D(M)}E_{\overline{\left \langle i,j\right\rangle}}(-1)\right ]\left [ \prod_{\left\langle i,j \right \rangle \in D(M)}E(i,j)\prod_{\left\langle i,j \right \rangle \in D(M)}E(i,j)\right ]\\
&=-I
\end{aligned},
\end{equation}
 thus $M^{-1}=-M$.

\item By Lemma~\ref{lemma4.2} and (i), $M^T=M^{-1}=-M$. If there was any non-zero element on the main diagonal of $M$, then $M^T\ne -M$, which causes a contradiction.

\item Assume matching operator $\overline{M}$ where $D(\overline{M})=\left \{  \overline {\left \langle i,j \right \rangle} |  \left \langle i,j \right \rangle \in D(M) \right \}$, then by Lemma~\ref{lemma4.1} and Eq.~(\ref{eq4}) we have 
\noindent
\begin{equation}
\begin{aligned}
 &-M=-\prod_{\left\langle i,j \right \rangle \in D(M)}E_{\left \langle i,j\right\rangle}(-1)\prod_{\left\langle i,j \right \rangle \in D(M)}E(i,j) \\
&=\left [ \prod_{\left\langle i,j \right \rangle \in D(M)}E_{\overline{\left \langle i,j\right\rangle}}(-1)\prod_{\left\langle i,j \right \rangle \in D(M)}E_{\left \langle i,j\right\rangle}(-1)\right ] \left [\prod_{\left\langle i,j \right \rangle \in D(M)}E_{\left \langle i,j\right\rangle}(-1)\prod_{\left\langle i,j \right \rangle \in D(M)}E(i,j)\right ] \\
&= \prod_{\left\langle i,j \right \rangle \in D(M)}E_{\overline{\left \langle i,j\right\rangle}}(-1)\prod_{\left\langle i,j \right \rangle \in D(M)}E(i,j) \\
&=\prod_{\overline{\left\langle i,j \right \rangle} \in D(\overline{M})}E_{\overline{\left \langle i,j\right\rangle}}(-1)\prod_{\overline{\left\langle i,j \right \rangle} \in D(\overline{M})}E(i,j) \\
&=\overline{M}
\end{aligned},
\end{equation}
thus $-M$ is also a matching operator.

\item For any scattered matrix $A$, by Lemma~\ref{lemma4.2} we have $A^{-1}=A^T$, thus its’ anti-self-reversibility and anti-symmetry are equivalent: $A^{-1}=A^T=-A$. Let $A$ be anti-symmetric, for every two anti-symmetrical elements in $A$, we exchange them, and multiply the row where $-1$ is by $-1$, then we get the identity matrix $I$. The above operations define $\frac{N}{2}$ couples and a set $D$ of these couples, then we have $\prod_{\left\langle i,j \right \rangle \in D}E_{\left \langle i,j\right\rangle}(-1)$ $\prod_{\left\langle i,j \right \rangle \in D}E(i,j)A=I$, thus $A^{-1}=\prod_{\left\langle i,j \right \rangle \in D}E_{\left \langle i,j\right\rangle}(-1)\prod_{\left\langle i,j \right \rangle \in D}E(i,j)$, which is a matching operator of which the division is $D$. By (iii), $A$ is also a matching operator $\Box$.
\end{romanlist}
\end{proof}

\subsection{Cooperation and matching equation}\label{subsec4.2}
\noindent
We give the following definitions.

\begin{definition}[\textbf{Cooperation}] 
Two matching operators $A,B$ are called \textbf{cooperative} if there is a new matching operator $C$ so that $B=CA$. By Proposition~\ref{prop4.1} we have $A=C^{-1}B$, thus this definition is symmetric for $A,B$.
\end{definition}

\begin{definition}[\textbf{Cooperative set}]\label{def4.6}
A set $\left \{ M_1, \cdots, M_{m} \right \}$ of $N$-order matching operators is called a \textbf{cooperative set} if any two of these operators are cooperative. A cooperative set $G=\left \{ I,M_1, \cdots, M_{N-1} \right \}$ is called \textbf{complete} (for convenience, $I$ is also regarded as a ``matching operator" but ignored).
\end{definition}

\begin{definition}[\textbf{Path length}]
For a row index $i$, we count the number of times $i$ was multiplied by $-1$ through its \textbf{path} in an operation $g$ (i.e., the mapping process of $i$), and call the value
of the number modulo $2$ the \textbf{length} of the path, denoted as $len_g(i)$. A length is actually a Boolean value indicating whether $i$ is multiplied by $-1$. e.g., let $g=BA$, then the path $i$ passes through in $g$ is $i \overset{\left\langle i,j \right \rangle _A \left \langle j,k \right \rangle _B}{\mapsto} (-1)^{len_g(i)}k$. According to the definition of a couple, $len_g(i)$ equals the XOR of all couples values on the path, i.e., $len_g(i)=\left \langle i,j \right \rangle _A \oplus \left \langle j,k \right \rangle _B$.
\end{definition}

\begin{remark}\label{rem4.2}
\begin{romanlist}
\item Given two matching operators $A,B$, then from vector $a$, two vectors $Aa, Ba$ orthogonal with $a$ can be obtained. If $Aa \perp Ba$ as well, then these vectors are pairwise orthogonal. By Proposition~\ref{prop3.1}, this means that vector $Ba$ is generated by performing a matching operator $C$ to vector $Aa$, i.e., $B=CA$. That’s why to define cooperation.

\item Given $N$ $N$-dimensional column vectors $\textbf{\textit{s}}^{(0)},\textbf{\textit{s}}^{(1)},\cdots,\textbf{\textit{s}}^{(N-1)}$ which constitute a special orthogonal matrix $S_N$. For the first vector $\textbf{\textit{s}}^{(0)}$, any other vector must be generated by performing a matching operator to it: $\textbf{\textit{s}}^{(i)}=M_i\textbf{\textit{s}}^{(0)}$, where $M_i$ is the matching operator of column $i$ to column $0$, and $M_0=I$. Therefore, the set $G_{S_N}=\left \{ I,M_1, \cdots, M_{N-1} \right \}$ must be a complete cooperative set. Regardless of $\textbf{\textit{s}}^{(0)}$, the matrix $S$ must be a special orthogonal matrix if its set of matching operators is cooperative.
\end{romanlist}
\end{remark}

We give the following condition of cooperation.

\begin{proposition}\label{prop4.2}
The necessary and sufficient condition for matching operators $A,B$ is as follows:
\begin{itemize}
\item $D(A)$ can be divided into several \textbf{$4$-tuples} (e.g., $\left \langle i,j \right \rangle\& \left \langle k,l\right \rangle \to \left \{ i,j,k,l \right \}$, shown in Definition~\ref{def5.4});

\item Operation $B$ corresponds to a bijection between $2$-tuples in any $4$-tuples (e.g., $\left \{ i,j \right \}\leftrightarrow \left \{ k,l \right \}$);

\item Let $g=BABA$, then for any row index $i$, $len_g(i)=1$.
\end{itemize}
\end{proposition}

The above proposition is symmetric to $A,B$, and the equation $len_g(i)=1$ is called a \textbf{matching equation}. To prove Proposition~\ref{prop4.2}, the following lemmas are required.

\begin{lemma}\label{lemma4.3}
$AB=-BA$ is a necessary and sufficient condition for the product $AB$ of anti-symmetric matrices $A,B$ to be anti-symmetric.
\end{lemma}

\begin{proof}
If $AB=-BA$, then $(AB)^T=B^TA^T=BA=-AB$. Conversely, if $(AB)^T=-AB$, then $AB=-(AB)^T=-B^TA^T=-BA$ $\Box$.
\end{proof}

\begin{lemma}\label{lemma4.4}
The following properties of cooperation hold:
\begin{romanlist}
\item The necessary and sufficient condition for matching operators $A,B$ to cooperate is that the product $AB$ or $BA$ is a matching operator.
\item There are not any two similar couples (i.e., they have the same row indexes, shown in Definition~\ref{def5.6}) between the divisions of matching operators $A,B$ which are cooperative.
\item There are nt any two similar couples in the divisions of any two matching operators in a cooperative set.
\end{romanlist}
\end{lemma}

\begin{proof}
\begin{romanlist}
\item By Proposition~\ref{prop4.1}, if $BA$ is a matching operator, so is $AB$ because $AB=-BA$, and vice versa. Therefore, we only need to prove in the case of $BA$ to be a matching operator, and then we have $B=-(BA)A$, which means that there is a matching operator $C=-BA$ so that $B=CA$, i.e., $A,B$ cooperate; Conversely, let $A,B$ cooperate, then there is a matching operator $C$ so that $B=CA$. Right multiply both sides of $B=CA$ by $A$ at the same time, then $BA=-C$, i.e., $BA$ is a matching operator.
\item By (i), $AB$ is a matching operator, in which any element $(i,i)=0$ by Proposition~\ref{prop4.1}. If there were any two similar couples between $D(A),D(B)$, e.g., $\left\langle i,j \right \rangle$, then there is a scattered point $(i,j)$ in $A$, $(j,i)$ in $B$, thus $(i,i)$ in $AB$, which causes a contradiction.
\item Obviously by (ii) and the definition of a cooperative set (see Definition~\ref{def4.6}) $\Box$.
\end{romanlist}
\end{proof}

Now we can prove Proposition~\ref{prop4.2}.

\begin{proof}
[Proof of Proposition~\upshape{\ref{prop4.2}}]
\begin{itemize}
\item (\textit{Sufficiency}) Without loss of generality, let $\left \langle i,j \right \rangle _A,\left \langle k,l \right \rangle _A \in D(A)$, $\left \langle j,k \right \rangle _B,\left \langle i,l \right \rangle _B \in D(B)$, and $\left \{ i,j,k,l \right \}$ be a $4$-tuple. Then for $i$, we can write down its path in $g=BABA$ as
\noindent
\begin{equation}\label{eq7}
i\overset{\left \langle i,j \right \rangle _A \left \langle j,k \right \rangle _B \left \langle k,l \right \rangle _A \left \langle l,i \right \rangle _B}{\mapsto} (-1)^{len_g(i)}i,
\end{equation}
thus 
\noindent
\begin{equation}\label{eq8}
len_g(i)=\left \langle i,j \right \rangle _A\oplus \left \langle j,k \right \rangle _B \oplus \left \langle k,l \right \rangle _A \oplus \left \langle l,i \right \rangle _B.
\end{equation}
And for $j$, its path is
\noindent
\begin{equation}\label{eq9}
j\overset{\left \langle j,i \right \rangle _A \left \langle i,l \right \rangle _B \left \langle l,k \right \rangle _A \left \langle k,j \right \rangle _B}{\mapsto} (-1)^{len_g(j)}j.
\end{equation}
Note that 
\noindent
\begin{equation}\label{eq10}
\begin{aligned}
&len_g(j)=\left \langle  j,i \right \rangle _A\oplus \left \langle i,l \right \rangle _B \oplus \left \langle l,k \right \rangle _A \oplus \left \langle k,j \right \rangle _B\\
&=\left(\left \langle i,j \right \rangle _A \oplus 1\right ) \oplus\left( \left \langle j,k \right \rangle _B\oplus 1\right)\oplus\left(  \left \langle k,l \right \rangle _A\oplus 1\right)\oplus\left(  \left \langle l,i \right \rangle _B\oplus 1\right) \\
&=\left \langle i,j \right \rangle _A\oplus \left \langle j,k \right \rangle _B \oplus \left \langle k,l \right \rangle _A \oplus \left \langle l,i \right \rangle _B\\
&=len_g(i)
\end{aligned},
\end{equation}
so we get $len_g(j)=len_g(i)=1$. Similarly, we have $len_g(k)=len_g(l)=1$, i.e., $\left \{i,j,k,l \right \}\overset{g}{\mapsto}\left \{ -i,-j,-k,-l\right \}$. The above conclusions are valid for all $4$-tuples, which means for any row $i$, $i \overset{g}{\mapsto} -i $, i.e., $BABA=-I$. By Lemma~\ref{lemma4.2}, $BA$ is a scattered matrix, thus is a matching operator by Proposition~\ref{prop4.1}, then $A,B$ cooperate by Lemma~\ref{lemma4.4}.

\item (\textit{Necessity}) Without loss of generality, for any couple $\left \langle j,k \right \rangle _B \in D(B)$, we know $\left \langle j,k \right \rangle \notin 	D(A)$ by Lemma~\ref{lemma4.4}. Suppose that $\left \langle i,j \right \rangle _A, \left \langle k,l \right \rangle _A\in D(A)$ correspondingly, and then $\left \langle u,i\right \rangle _B\in D(B)$, where $u$ is an undetermined index. We know that $BA$ is a matching operator by Lemma~\ref{lemma4.4}, so we have $g=BABA=-I$ by Proposition~\ref{prop4.1}. Therefore, for any row index $i$, we have $i\overset{g}{\mapsto} -i$, so $len_g(i)=1$. The specific path in $g$ is 
\noindent
\begin{equation}\label{eq11}
i\overset{\left \langle i,j \right \rangle _A \left \langle j,k \right \rangle _B \left \langle k,l \right \rangle _A \left \langle u,i \right \rangle _B}{\mapsto} (-1)^{len_g(i)}i,
\end{equation}
note that $u=l$ must hold, otherwise the path is not connected. Thus, we get a $4$-tuple $\left \{ i,j,k,l \right \}$, and $B$ is corresponding to a bijection $\left \{ i,j\right \} \leftrightarrow \left \{ k,l \right \}$. The above discussion applies to any couple in $D(B)$, so we can deduce the condition $\Box$.
\end{itemize}
\end{proof}

\subsection{Algorithm to generate a special orthogonal matrix}\label{subsec4.3}
\noindent
Based on Proposition~\ref{prop4.2}, we have the following proposition which allows us to reduce Problem~\ref{pro2} to the problem of a solving linear equation system.

\begin{proposition}\label{prop4.3}
Given an $N$-order semi-orthogonal matrix, the necessary and sufficient condition for it to have any solution is that the matching equation system combined from all matching equations about it has a solution.
\end{proposition}

\begin{proof} It can be seen from the proof process of Proposition~\ref{prop4.2} that for any $4$-tuple, the matching equations of the $4$ row indexes of it are equivalent to each other, so one $4$-tuple only corresponds to one independent matching equation. Therefore, there are $\frac{N}{4}$ independent matching equations between any two $N$-order matching operators satisfying the bijective condition in Proposition~\ref{prop4.2} (i.e., (a) and (b)), and the necessary and sufficient condition for these two matching operators to cooperate is that the equation system combined from these $\frac{N}{4}$ matching equations has a solution. For a semi-orthogonal matrix, it must satisfy the bijective condition similar to the proof of Proposition~\ref{prop4.2}, so this proposition can be proved $\Box$.
\end{proof}

\begin{remark}\label{rem4.3}
An $N$-order semi-orthogonal matrix corresponds to $N-1$ matching operators, and there are $\frac{N}{4}$ independent matching equations between every two matching operators, thus the above equation system has $\frac{N}{4}\binom{N-1}{2}=\frac{(N-1)(N-2)N}{8} $ equations, and there are $\frac{N(N-1)}{2}$ unknowns which is the total number of all couples. Therefore, we can get an $R\times C$ augmented matrix $\overline{A}$, where $R=\frac{(N-1)(N-2)N}{8}$ and $C=\frac{N(N-1)}{2}+1$. 
\end{remark}

We can use XOR Gaussian elimination to determine the existence of solutions of an $N$-order semi-orthogonal matrix, to solve Problem~\ref{pro2}. We have presented Algorithm~\ref{algo1} (see Appendix~\ref{secA} for details). The algorithm only solves Problem~\ref{pro2}, but because it is very simple to generate a semi-orthogonal matrix (as described by Wei et al.\cite{2018WeiOpt}), it can be regarded as an almost complete algorithm to generate a special orthogonal matrix. In addition, the algorithm can also be used to determine the existence of solutions of a semi-orthogonal matrix, which can be used in the following section.

\begin{algorithm}
\caption{Find a solution of an $N$-order semi-orthogonal matrix. }\label{algo1}
\begin{algorithmic}[1]
\REQUIRE An $N$-order matrix $S$ where $N=2^n$ and $n >0$.
\ENSURE  $S$ is semi-orthogonal and constructed by $U=[N]$.
\STATE Compute all divisions of matching operators of $S$ and store them in table $T$.
\STATE Set an $R\times C$ zero Boolean matrix $\overline{A}$ as an augmented matrix.
\STATE Input matching equations into $\overline{A}$ for every two matching operators stored in $T$.
\STATE Use XOR Gaussian elimination on $\overline{A}$.
\IF{$rank(\overline{A})=rank(A)$}
    \STATE Generate a special solution $X$ of the matching equation.
    \STATE Generate a special orthogonal matrix corresponding to $S$ by $X$.
    \STATE Return $TRUE$.
\ELSE
    \STATE Return $FALSE$.
\ENDIF
\end{algorithmic}
\end{algorithm}

We analyze the complexity of Algorithm~\ref{algo1} in Appendix~\ref{secA}-(3). The exact time and space complexity of Algorithm~\ref{algo1} are $O(N^9)$ and $O(N^5)$ respectively. Obviously, the algorithm is polynomial for $N$.

\section{Infeasibility proof of constructing a special orthogonal matrix for the DRSP of an arbitrary $n$-qubit state}\label{sec5}
\noindent
In this section, we will prove the infeasibility of constructing a special orthogonal matrix, i.e., the non-solvability of Problem~\ref{pro1}. The main idea is to prove that each semi-orthogonal matrices can be simplified into a unique form (called ordered type), and then use Algorithm~\ref{algo1} to determine that the unique form of order $16$ or larger has no solution. The proof is arranged as follows: in Section~\ref{subsec5.1}, we define the semi-matching operator and semi-cooperation and prove that an $N$-size semi-cooperative set of $N$-order semi-matching operators is a group. In Section~\ref{subsec5.2}, we study $n$-bijections so-called and prove that the generator set of a semi-cooperative group satisfies a property about $n$-bijections. We prove that any mapping table of the generator set of a semi-cooperative group can be ordered in Section~\ref{subsec5.3}, and then in Section~\ref{subsec5.4}, we prove that any semi-orthogonal matrix can be simplified into the ordered type based on it. In the end, we prove that for $n>3$ the ordered type is non-solvable by using Algorithm~\ref{algo1}, so as to prove the infeasibility in Section~\ref{subsec5.5}.

\subsection{Semi-matching operators and semi-cooperative group}\label{subsec5.1}
\noindent
For semi-orthogonal matrices, we only need to study its exchanging operation, so we have the following definitions.

\begin{definition}[\textbf{Semi-matching operator}]
Similar to Eq.~(\ref{eq4}), a \textbf{semi-matching operator} $M$ is as follows:
\noindent
\begin{equation}\label{eq5.1}
M=\prod_{\left\langle i,j \right \rangle \in D(M)}E(i,j),
\end{equation}which is the operator corresponding to an operation of pairwise row-exchanging. Its couples such as $\left\langle i ,j \right\rangle$ and division $D(M)$ can also be defined similarly.
\end{definition}

\begin{definition}[\textbf{Semi-cooperation}]
Two semi-matching operators $A,B$ are called \textbf{semi-cooperative} if there is a new semi-matching operator $C$ so that $B=CA$.
\end{definition}

\begin{definition}[\textbf{Semi-cooperative set}]
A set $\left \{ M_1, \cdots, M_{m} \right \}$ of $N$-order semi-matching operators is called a \textbf{semi-cooperative set} if any two of these operators are semi-cooperative. A semi-cooperative set $G=\left \{I, M_1, \cdots, M_{N-1} \right \}$ is called \textbf{complete}.
\end{definition}

The following proposition is a magical property of a semi-cooperative set.

\begin{proposition}\label{prop5.1}
A complete semi-cooperative set $G=\left \{ I,M_1, \cdots, M_{N-1} \right \}$ of $N$-order operators is an Abel group under matrix multiplication.
\end{proposition}

To prove Proposition~\ref{prop5.1}, we first need the following lemmas.

\begin{lemma}\label{lemma5.1}
The following properties of a semi-matching operator $M$ hold:
\begin{romanlist}
\item (\textit{Self-reversibility}) $M^{-1}=M$;
\item (\textit{Symmetry}) $M^T=M$, and all elements on the main diagonal of $M$ are $0$;
\item A scattered matrix $M$ must be a semi-matching operator if of which all elements on the main diagonal are $0$, and it’s self-reversible or symmetric.
\item The cooperation between semi-matching operators $A,B$ is equivalent to that the product $AB$ or $BA$ is also a semi-matching operator, and to that $AB=BA$.
\item There are not any two similar couples between the divisions of two semi-cooperative semi-matching operators $A,B$, and thus in the divisions of any two operators in a semi-cooperative set.
\item The necessary and sufficient condition for semi-matching operators $A,B$ to cooperate is that $D(A)$ can be divided into several $4$-tuples, and operation $B$ corresponds to a bijection between $2$-tuples in any $4$-tuples.
\end{romanlist}
\end{lemma}

\begin{proof} 
Similar to Proposition~\ref{prop4.1}, Proposition~\ref{prop4.2}, Lemma~\ref{lemma4.3} and Lemma~\ref{lemma4.4} $\Box$.
\end{proof}

\begin{lemma}\label{lemma5.2}
For a complete semi-cooperative set of $N$-order operators, $\forall 0\le i < N$, we denote $M^{(i)}$ as the operator satisfying $\left \langle 0,i \right \rangle\in D(M^{(i)})$. Then $\forall 0\le i,j <N$, if $\left \langle j,x \right \rangle\in D(M^{(i)})$, then $\left \langle i,x \right \rangle\in D(M^{(j)})$.
\end{lemma}

\begin{proof}
Because $M^{(i)}$ semi-cooperate with $M^{(j)}$ and $\left \langle 0,i \right \rangle , \left \langle j,x \right \rangle\in D(M^{(i)})$, then $\left\{ 0,i,j,x \right \}$ is a $4$-tuple and $M^{(j)}$ corresponds to a bijection between $\left \langle 0,i \right \rangle , \left \langle j,x \right \rangle$ by Lemma~\ref{lemma5.1}. By $\left \langle 0,j \right \rangle\in D(M^{(j)})$, we have $\left \langle i,x \right \rangle\in D(M^{(j)})$ $\Box$.
\end{proof}

\begin{remark}
By Lemma~\ref{lemma5.1}, there are not any two similar couples in $G$, thus all $(N-1)\frac{N}{2}$ couples in $G$ traverse exactly all possible $\binom{N}{2}=\frac{N(N-1)}{2}$ couples. Thus, $\forall 0\le i < N$, $\left \langle 0,i \right \rangle$ must exist in $G$, and for $i\ne i'$, $\left \langle 0,i \right \rangle$ and $\left \langle 0,i' \right \rangle$ belong to different operators. 
\end{remark}

Now we can prove Proposition~\ref{prop5.1}.

\begin{proof}[Proof of Proposition~\upshape{\ref{prop5.1}}] Obviously the associativity, reversibility ($M^{-1}=M$), unit element existence ($I$), and commutativity ($AB=BA$ if they semi-cooperate) of $G$ are all satisfied. For the closeness, we only need to prove that for any two operators $M^{(a)},M^{(b)}\in G$, $M^{(b)}M^{(a)}\in G$ also holds, where $a,b$ is as specified in Lemma~\ref{lemma5.2}.

Assume $\left \langle b,c \right \rangle_a \in D(M^{(a)})$ and consider $M^{(c)}\in G$. Because of the semi-cooperation between $M^{(a)},M^{(b)}$, we can consider each $4$-tuple of them. Without loss of generality, assume $\left \langle i,j \right \rangle_a \left \langle k,l \right \rangle_a \in D(M^{(a)}) ,\left \langle i,k \right \rangle_b \left \langle j,l \right \rangle_b \in D(M^{(b)})$, then $\left \langle i,l \right \rangle_{ba} \left \langle j,k \right \rangle_{ba} \in D(M^{(b)}M^{(a)})$. Consider $M^{(c)}$, then for $\left \langle i,j \right \rangle_a $, $M^{(c)}$ corresponds to a bijection from it to a new couple $\left \langle x,y \right \rangle_a \in D(M^{(a)})$, and we assume $\left \langle i,x \right \rangle_c, \left \langle j,y \right \rangle_c \in D(M^{(c)})$. Consider $M^{(i)}$, then we have $\left \langle b,k \right \rangle_i,\left \langle c,x \right \rangle_i \in D(M^{(i)})$ because $\left \langle i,k \right \rangle_b \in D(M^{(b)}), \left \langle i,x \right \rangle_c \in D(M^{(c)})$ by Lemma~\ref{lemma5.2}. Note that $M^{(i)}$ and $M^{(a)}$ semi-cooperate and $\left \langle b,c \right \rangle_a\in D(M^{(a)})$, then $M^{(i)}$ have a bijection $\left \langle b,c \right \rangle_a \leftrightarrow \left \langle k,x \right \rangle _a$, which means $x=l$ thus $y=k$. Therefore,  $\left \langle i,l \right \rangle_{c}\left\langle j,k \right \rangle_c \in D(M^{(c)})$, which is the same as $M^{(b)}M^{(a)}$. 

For each $4$-tuple, $M^{(c)}$ is the same as $M^{(b)}M^{(a)}$, thus $M^{(b)}M^{(a)}=M^{(c)}\in G$. Consequently, the set is an Abel group $\Box$.
\end{proof}

Considering the self-reversibility,  as an Abel group of order $N=2^n$, a semi-cooperative group $G$ of $N$-order semi-matching operators has $n$ independent generators $\left \{ M_1, M_2,\cdots, M_{n} \right \}$, which form a generator set $\hat{G}$ of size $n$. Proposition~\ref{prop5.1} indicates that if we have a generator set of the group, then the group is determined. Thus, Problem~\ref{pro3} can be reduced to the following problem.

\begin{problem}\label{pro4}(\textbf{Find all semi-cooperative generator sets})
Given $N=2^n$, find all $n$-size generator sets of $N$-order semi-matching operators.
\end{problem}

\subsection{Bijections and semi-cooperative generator set}\label{subsec5.2}
\noindent
Observe the following semi-orthogonal matrices:
\noindent
\begin{equation}\label{eq12}
 \left | \overline{S}_4 \right |=\begin{pmatrix}
 a_0 & a_1 & a_2 & a_3\\
 a_1 & a_0 & a_3 & a_2\\
 a_2 & a_3 & a_0 & a_1\\
 a_3 & a_2 & a_1 & a_0
\end{pmatrix},
\left | \overline{S}_8 \right |=\begin{pmatrix}
 a_0 & a_1 & a_2 & a_3 & a_4 & a_5 & a_6 & a_7\\
 a_1 & a_0 & a_3 & a_2 & a_5 & a_4 & a_7 & a_6\\
 a_2 & a_3 & a_0 & a_1 & a_6 & a_7 & a_4 & a_5\\
 a_3 & a_2 & a_1 & a_0 & a_7 & a_6 & a_5 & a_4\\
 a_4 & a_5 & a_6 & a_7 & a_0 & a_1 & a_2 & a_3\\
 a_5 & a_4 & a_7 & a_6 & a_1 & a_0 & a_3 & a_2\\
 a_6 & a_7 & a_4 & a_5 & a_2 & a_3 & a_0 & a_1\\
 a_7 & a_6 & a_5 & a_4 & a_3 & a_2 & a_1 & a_0
\end{pmatrix},
\end{equation}
note that compared to column $0$, column $1$ has row-exchanging: $\{0\}\leftrightarrow \{1\},\{2\}\leftrightarrow \{3\},\cdots$; column $2, 3$ have row-exchanging: $\{0,1\}\leftrightarrow \{2,3\},\{4,5\}\leftrightarrow \{6,7\}$; columns $4, 5, 6, 7$ have row-exchanging $\{0,1,2,3\}\leftrightarrow \{4,5,6,7\}$. The above property is abstracted as follows: a semi-cooperative set $\{ M_1,M_2,\cdots, M_n\}$ of $N$-order semi-matching operators satisfies:
\begin{itemize}
\item $M_1$ has row-exchanging: $\{i\}\leftrightarrow \{j\},\{k\}\leftrightarrow \{l\},\cdots$
\item $M_2$ has row-exchanging: $\{i,j\}\leftrightarrow \{k,l\},\{m,n\}\leftrightarrow \{o,p\},\cdots$
\item $M_3$ has row-exchanging: $\{i,j,k,l\}\leftrightarrow \{m,n,o,p\},\{q,r,s,t\}\leftrightarrow \{u,v,w,x\},\cdots$
\item $\cdots$
\end{itemize}

We will strictly describe and prove this property. Firstly, take any positive integer power $n$ of $2$, where $1\le n\le N$, we have the following definitions.

\begin{definition}[\textbf{$n$-tuple}]\label{def5.4}
A set $\left \{ k_1,k_2,\cdots.k_n\right \}$ of  row (or column) indexes is called an \textbf{$n$-tuple}, denoted as $T^n$, which is the promotion of a couple or a $4$-tuple.
\end{definition}

\begin{definition}[\textbf{$n$-bijection}]
An \textbf{$n$-bijection} satisfies the following recursive definition:
\begin{itemize}
\item A $1$-bijection $\sigma^1$ is a bijection between two $1$-tuples, i.e., $\sigma^1:T^1_i \leftrightarrow T^1_j$.
\item An $n$-bijection $\sigma^n$ is a bijection between two $n$-tuples: $\sigma^n:\left (T^{\frac{n}{2}}_i+ T^{\frac{n}{2}}_j \right ) \leftrightarrow \left (T^{\frac{n}{2}}_k+ T^{\frac{n}{2}}_l\right )$, where ``$+$" means taking the union set. We stipulate that $\sigma^n=\left \{T^{\frac{n}{2}}_i \leftrightarrow T^{\frac{n}{2}}_l ,T^{\frac{n}{2}}_j \leftrightarrow T^{\frac{n}{2}}_k\right \} $ or $\left \{T^{\frac{n}{2}}_i \leftrightarrow T^{\frac{n}{2}}_k ,T^{\frac{n}{2}}_j \leftrightarrow T^{\frac{n}{2}}_l\right \} $, where $T^{\frac{n}{2}}_i \leftrightarrow T^{\frac{n}{2}}_l$ and others are all $\frac{n}{2}$-bijections.
\end{itemize}
\end{definition}

\begin{definition}[\textbf{$n$-similar}]\label{def5.6}
Two $n$-bijections $\sigma^n_i:T^n_{i1} \leftrightarrow  T^n_{i2},\sigma^n_j:T^n_{j1} \leftrightarrow T^n_{j2}$ are called \textbf{similar}, denoted as $\sigma^n_i \sim \sigma^n_j$, if their image and preimage sets are the same respectively, i.e., $T^n_{i1}=T^n_{j1}$ and $T^n_{i2}=T^n_{j2}$, and they are combined from the same $\frac{n}{2}$-bijections. Two $N$-order semi-matching operators $A,B$ are called \textbf{$n$-similar}, denoted as $A \overset{n}{\sim }B$, if each $n$-bijection $\sigma^n_i$ of $A$ one-to-one corresponds to an $n$-bijection $\sigma^n_j$ of $B$ so that $\sigma^n_i \sim \sigma^n_j$, and called \textbf{$n$-dissimilar} if not, denoted as $A \overset{n}{\nsim }B$. $n$-similarity is obviously an equivalence relation.
\end{definition}

\begin{remark}
\begin{romanlist}
\item Obviously, each $n$-bijection $\sigma^n$ can be written as $\sigma^n=T^n_i \leftrightarrow T^n_j$, and even can be a self-mapping, i.e., $T^n_i=T^n_j$.

\item Any $n$-bijection $\sigma^n=\left \{T^{\frac{n}{2}}_i \leftrightarrow T^{\frac{n}{2}}_l ,T^{\frac{n}{2}}_j \leftrightarrow T^{\frac{n}{2}}_k\right \} $ can be seen as a combination (i.e., product) of two $\frac{n}{2}$-bijections:  $\sigma^n=\sigma^{\frac{n}{2}}_1 \sigma^{\frac{n}{2}}_2$, where 
$\sigma^{\frac{n}{2}}_1:T^{\frac{n}{2}}_i \leftrightarrow T^{\frac{n}{2}}_l ,\sigma^{\frac{n}{2}}_2:T^{\frac{n}{2}}_j \leftrightarrow T^{\frac{n}{2}}_k$.

\item For any $n$-bijection $\sigma^n=T^n_i \leftrightarrow T^n_j$ that is not a self-mapping, there is only one $2n$-tuple $T^{2n}=T^n_i \cup T^n_j$ corresponding to it. Without confusion, we can consider $\sigma^n$ as $T^{2n}$, e.g., a couple is both a $1$-bijection and a $2$-tuple. Let semi-matching operator $A$ has two $\frac{n}{2}$-bijections $T^{\frac{n}{2}}_i \leftrightarrow T^{\frac{n}{2}}_k ,T^{\frac{n}{2}}_j \leftrightarrow T^{\frac{n}{2}}_l $, semi-matching operator $B$  has $\left (T^{\frac{n}{2}}_i+ T^{\frac{n}{2}}_j \right ) \leftrightarrow \left (T^{\frac{n}{2}}_k+ T^{\frac{n}{2}}_l\right )$, then we say $B$ performs on an $n$-tuple $\left (T^{\frac{n}{2}}_i+ T^{\frac{n}{2}}_j \right )$  of $A$.

\item Similar $n$-bijections can be called of the same class. Any class of $n$-bijection $\sigma^n:\left (T^{\frac{n}{2}}_i+ T^{\frac{n}{2}}_j \right ) \leftrightarrow \left (T^{\frac{n}{2}}_k+ T^{\frac{n}{2}}_l\right )$ corresponds to $2$ classes of specific $n$-bijection:  $\sigma^n_1=\left \{T^{\frac{n}{2}}_i \leftrightarrow T^{\frac{n}{2}}_l ,T^{\frac{n}{2}}_j \leftrightarrow T^{\frac{n}{2}}_k\right \} $ and $\sigma^n_2\left \{T^{\frac{n}{2}}_i \leftrightarrow T^{\frac{n}{2}}_k ,T^{\frac{n}{2}}_j \leftrightarrow T^{\frac{n}{2}}_l\right \}$. 
\end{romanlist}
\end{remark}

Now we give the following sufficient condition for the property proposed at the beginning.

\begin{proposition}\label{prop5.2}
The following properties about an $m$-size generator subset $\{ M_1,M_2,\cdots, M_m\}$ hold:
\begin{romanlist}
\item In the subgroup $\left \langle M_1,M_2,\cdots, M_m \right \rangle$ generated by $\{ M_1,M_2,\cdots, M_m\}$, all operators are pairwise $2^m$-similar.
\item $\forall 1\le j< i \le m$, $M_i$ performs on each $2^j$-tuple of $M_j$.
\end{romanlist}
\end{proposition}

The above proposition means that a sufficient condition of the property at the beginning of this section is that the set is a generator subset. To prove it, we first need the following lemmas.

\begin{lemma}\label{lemma5.3}
The following properties of $n$-bijection hold:
\begin{romanlist}
\item (\textit{Non-overlapping}) In an $N$-order semi-matching operator, there are not any two different overlapping $n$-bijection (i.e., their image or preimage sets are the same);
\item (\textit{Boundedness}) In a specific complete semi-cooperative set of $N$-order semi-matching operators, the number of any class of $n$-bijections is $n$ at most.
\end{romanlist}
\end{lemma}

\begin{proof}
We induct on $n$.
\begin{romanlist}
\item If there are any two overlapping $n$-bijections, then at least one row index is mapped into two indexes, which can not be by a matching operator.

\item {\begin{alphlist}
        \item By Lemma~\ref{lemma5.1}, the number of any class of $1$-bijection is $1$ at most.
        \item Assume we have proved that the number of any class of $n$-bijection is $n$ at most. For any class of $2n$-bijections $\sigma^{2n}:\left (T^n_i + T^n_j\right ) \leftrightarrow \left (T^n_k + T^n_l\right )$, if it exists in the set, then it corresponds to $2$ classes of specific $2n$-bijections:  $\sigma^{2n}_1=\left \{T^n_i \leftrightarrow T^n_l ,T^n_j \leftrightarrow T^n_k\right \} ,\sigma^{2n}_2=\left \{T^n_i \leftrightarrow T^n_k ,T^n_j \leftrightarrow T^n_l\right \} $, which are both combined from $2$ classes of $n$-bijections. Each possibility of $T^n_i \leftrightarrow T^n_l$ always exists with one possibility of $T^n_j \leftrightarrow T^n_k$, thus if the number of $T^n_i \leftrightarrow T^n_l$ is $n$, so is $T^n_j \leftrightarrow T^n_k$. Therefore, the number of $\sigma^{2n}_1$ is at most $n$, and so is $\sigma^{2n}_2$, so the number of $\sigma^{2n}$ is $2n$ at most (We are not to traverse all theoretical possibilities of $\sigma^{2n}_1,\sigma^{2n}_2$, but consider in a specific complete semi-cooperative set, in which the number is far less than the former).
		\end{alphlist}}The induction is completed $\Box$.
\end{romanlist}
\end{proof}

\begin{lemma}\label{lemma5.4}
Have two $N$-order semi-matching operators $A,B$, the following properties of $n$-similarity hold:
\begin{romanlist}
\item If $A \overset{n}{\sim }B$, then $A \overset{2n}{\sim }B$, and thus $A \overset{n2^k}{\sim }B$ ($n2^k\le N$, where $k$ is a positive integer).
\item If $A \overset{n}{\nsim }B$, then $A \overset{\frac{n}{2}}{\nsim }B$, and thus $A \overset{n2^{-k}}{\nsim }B$ ($n2^{-k}\ge 1$, where $k$ is a positive integer).
\end{romanlist}
\end{lemma}

\begin{proof}
\begin{romanlist}
\item Each $2n$-bijection $\sigma^{2n}_i$ of $A$ can be written as a product of two $n$-bijections $\sigma^n_{i1},\sigma^n_{i2}$, i.e., $\sigma^{2n}_i=\sigma^n_{i1}\sigma^n_{i2}$. Because $A \overset{n}{\sim }B$, there are corresponding $\sigma^n_{j1} \sim\sigma^n_{i1},\sigma^n_{j2}\sim \sigma^n_{i2}$ of $B$. By Lemma~\ref{lemma5.3}, $\sigma^n_{j1},\sigma^n_{j2}$ don’t overlap, so $\sigma^{2n}_j=\sigma^n_{j1}\sigma^n_{j2}\sim\sigma^n_{i1}\sigma^n_{i2}=\sigma^{2n}_i$, i.e., $\sigma^{2n}_j$ is the $n$-bijection of $B$ similar to $\sigma^{2n}_i$, thus $A \overset{2n}{\sim }B$. We can see $A \overset{n2^k}{\sim }B$ by continuous recursion.
\item Similar to (i) $\Box$.
\end{romanlist}
\end{proof}

Now we can prove Proposition~\ref{prop5.2}.

\begin{proof}[Proof of Proposition~\upshape{\ref{prop5.2}}]
We induct on $m$.
\begin{romanlist}
\item {Note that for any two semi-matching operators $A,B$, if assume $\left \langle i,j \right \rangle_A \left \langle k,l \right \rangle_A \in D(A)$, $\left \langle i,k \right \rangle_B \left \langle j,l \right \rangle_B \in D(B)$, then $\left \langle i,l \right \rangle_{BA} \left \langle j,k \right \rangle_{BA} \in D(BA)$. Thus $A,B,BA$ are pairwise $4$-similar.
\begin{alphlist}
\item Obliviously, in $\left \langle M_1 \right \rangle=\{I,M_1\}$ all operators are pairwise $2^1=2$-similar.
\item For $m>1$, Assume that we have proved that in $\left \langle M_1, M_2, \cdots, M_m \right \rangle$ all operators are pairwise $2^m$-similar, thus $2^{m+1}$-similar by Lemma~\ref{lemma5.4}. For $M_{m+1}$, $\forall \overline{M}_i \in \left \langle M_1, M_2, \cdots, M_m \right \rangle$, we have $M_{m+1}, \overline{M}_i, M_{m+1}\overline{M}_i$ are pairwise $4$-similar, thus $2^{m+1}$-similar. Note that $\left \langle M_1, M_2, \cdots, M_m,M_{m+1} \right \rangle=\left \langle M_1, M_2, \cdots, M_m \right \rangle +  \{ M_{m+1}\overline{M}_i | \overline{M}_i \in \left \langle M_1, M_2, \cdots, M_m \right \rangle \}$. Due to the transitivity of similarity, the $2^{m+1}$-similarity holds for $\left \langle M_1, M_2, \cdots, M_{m+1} \right \rangle$.
\end{alphlist}}

\item {By (i), $\forall 1\le j <i \le m$, in $\left \langle M_1, M_2, \cdots, M_{j-1} \right \rangle$ of which the order is $2^{j-1}$, all operators are pairwise $2^{j-1}$-similar, thus all $2^{j-1}$ possibilities of these $2^{j-1}$-bijections are traversed by Lemma~\ref{lemma5.3}. Therefore, any additional generator $M_i$ is $2^{j-1}$-non-similar, thus $2^{j-2}$-non-similar to them.
\begin{alphlist}
\item Obliviously for $\forall 1\le j< i \le 1$.
\item Assume that we have proved that $\forall 1\le j< i\le m$, $M_i$ performs on each $2^j$-tuple of $M_j$. Consider $M_{m+1}$. For $j=1$ it is obvious by Lemma~\ref{lemma5.1}. $\forall 1< j \le m$, $M_{m+1}$ cooperates with $M_{j-1}$, respectively, thus it performs on each $2^{j-1}$-tuple of $M_{j-1}$ as well as $M_j$. By $M_{m+1} \overset{2^{j-1}}{\nsim }M_j$ and similar to the proof of Proposition~\ref{prop5.2}, without loss of generality, suppose that there are $T^{2^{j-1}}_1 \leftrightarrow_j T^{2^{j-1}}_2,T^{2^{j-1}}_3 \leftrightarrow_j T^{2^{j-1}}_4$ in $2^{j-1}$-bijections of $M_j$, and $T^{2^{j-1}}_1 \leftrightarrow_{m+1} T^{2^{j-1}}_3,T^{2^{j-1}}_2 \leftrightarrow_{m+1} T^{2^{j-1}}_5$ of $M_{m+1}$. Because $M_{m+1},M_j$ semi-cooperate, we have mapping $T^{2^{j-1}}_1 \overset{M_{m+1}M_jM_{m+1}M_j}{\to} T^{2^{j-1}}_1$, i.e., $T^{2^{j-1}}_1 \leftrightarrow_j T^{2^{j-1}}_2 \leftrightarrow_{m+1} T^{2^{j-1}}_5 \leftrightarrow_j T^{2^{j-1}}_3 \leftrightarrow_{m+1} T^{2^{j-1}}_1$, leading to $ T^{2^{j-1}}_5 \leftrightarrow_j T^{2^{j-1}}_3$. However, there is $T^{2^{j-1}}_3 \leftrightarrow_j T^{2^{j-1}}_4$ of $M_j$, thus $T^{2^{j-1}}_5 =T^{2^{j-1}}_4$ must hold because of $2^{j-1}$-non-overlapping by Lemma~\ref{lemma5.3}. Therefore, there is $2^j$-bijection $\left (T^{2^{j-1}}_1 + T^{2^{j-1}}_2 \right ) \leftrightarrow \left (T^{2^{j-1}}_3 + T^{2^{j-1}}_4\right )$ of $M_{m+1}$, where $\left (T^{2^{j-1}}_1 + T^{2^{j-1}}_2 \right )$, $\left (T^{2^{j-1}}_3 + T^{2^{j-1}}_4 \right )$ are $2^j$-tuples of $M_j$, i.e., $M_{m+1}$ performs on each $2^j$-tuple of $M_j$. Now we know it holds for $\forall 0\le j< i\le m+1$, so the induction is completed $\Box$.
\end{alphlist}}
\end{romanlist}
\end{proof}

\subsection{Mapping table and its ordering}\label{subsec5.3}
\noindent
To describe how we order a semi-cooperative generator set, we first need the following definitions.

\begin{definition}[\textbf{Mapping table}]
For a generator set $\hat{G}=\left \{ M_1, M_2,\cdots, M_{n} \right \}$, a mapping table is a permutation of index sequence $[N]$ in which $\forall M_m \in\hat{G}$, for each $2^m$-tuple of $M_m$, the two $2^{m-1}$-tuples in it are adjacent. For each $2^m$-tuple $T^{2^m}$ in a mapping table, the leftmost element $i$ of it is denoted as $i=left\left (T^{2^m}\right )$.
\end{definition}

\begin{definition}[\textbf{$n$-order}]
For a mapping table $T$ of a generator set $\hat{G}=\left \{ M_1, M_2,\cdots, M_{n} \right \}$, we have the following recursive definition:
{\begin{alphlist}
\item $T$ is directly called \textbf{$1$-ordered};

\item If $T$ is $(m-1)$-ordered, and for each $2^{m-1}$-bijection $\sigma: T^{2^{m-1}}_1 \leftrightarrow T^{2^{m-1}}_2$of $M_m$, $\sigma\left( left \left( T^{2^{m-1}}_1\right)\right)=left\left( T^{2^{m-1}}_2 \right)$, then $T$ is \textbf{$m$-ordered}.
\end{alphlist}}
\end{definition}

\begin{algorithm}
\caption{ Generate an $n$-ordered mapping table $T$ of a generator set $\hat{G}=\left \{ M_1, M_2,\cdots, M_{n} \right \}$ of $N=2^n$-order operators.}\label{algo2}
\begin{algorithmic}[1]
\REQUIRE A generator set $\hat{G}=\left \{ M_1, M_2,\cdots, M_{n} \right \}$ of $N=2^n$-order semi-matching operators and an $N$-size empty vector $T$.
\STATE $T\Leftarrow [N]$.
\FOR{$m \Leftarrow 1$ \textbf{to} $n$}
 	\FOR {$2^m$-tuple $T_x^{2^m}=T_{x1}^{2^{m-1}}+ T_{x2}^{2^{m-1}}$ of $M_m$}
    	\STATE Arrange $T_{x1}^{2^{m-1}}$ and $T_{x2}^{2^{m-1}}$ adjacent while maintaining their internal arrangement.
     \ENDFOR
     \FOR {$2^m$-tuple $T_x^{2^m}=T_{x1}^{2^{m-1}}+ T_{x2}^{2^{m-1}}$ of $M_m$}
    	\STATE $i\Leftarrow left\left( T_{x1}^{2^{m-1}} \right )$.
        \STATE $j\Leftarrow M_m(i)$.
        \IF {$j\ne left\left( T_{x2}^{2^{m-1}} \right )$}
        	\STATE $k\Leftarrow 2^{m-1}$.
            \STATE $R\Leftarrow T_{x2}^{2^{m-1}} $.
            \WHILE{$k>1$}
            	\STATE $R=T^{\frac{k}{2}}_1+T^{\frac{k}{2}}_2$.
                \IF {$j\in T^{\frac{k}{2}}_2$}
                	\FOR {$2^k$-tuple $T^k_y=T^{\frac{k}{2}}_{y1}+T^{\frac{k}{2}}_{y2} \subseteq T_{x2}^{2^{m-1}}$}
                    	\STATE Exchange $T^{\frac{k}{2}}_{y1}$ and $T^{\frac{k}{2}}_{y2}$ while maintaining their internal arrangement.
                    \ENDFOR
                    \STATE $k\Leftarrow \frac{k}{2}$.
                    \STATE $R\Leftarrow$ Current $T^{\frac{k}{2}}_1$.
                \ENDIF
            \ENDWHILE
        \ENDIF
     \ENDFOR
\ENDFOR
\end{algorithmic}
\end{algorithm}

Algorithm~\ref{algo2} shows how to generate an $n$-ordered mapping table. We have the following proposition.

\begin{proposition}\label{prop5.3}
The outcome $T$ of Algorithm~\ref{algo2} must be an $n$-ordered mapping table of $\hat{G}=\left \{ M_1, M_2,\cdots, M_{n} \right \}$.
\end{proposition}

To prove it, we first need the following lemma.

\begin{lemma}\label{lemma5.5}
Have a mapping table $T$ of a generator set $\hat{G}=\left \{ M_1, M_2,\cdots, M_{n} \right \}$, then $\forall 1\le m \le n$, if $T$ is $m$-ordered, then the couples of $\left \{ M_1, M_2, \cdots, M_m\right \}$ are determined.
\end{lemma}

\begin{proof}
We induct on $m$.
\begin{romanlist}
\item Obviously, for $m=1$, the couples of $\left \{ M_1\right\}$ are determined.
\item Assume that we have proved that for $1 \le m \le n-1$, the couples of $\left \{ M_1, M_2, \cdots, M_m\right \}$ are determined. For $m+1$, we need to prove that $M_{m+1}$ is determined. As a mapping table, $T$ tells us how $M_{m+1}$ perform on $2^m$-tuples of $M_m$. 

For each $2^{m+1}$-tuple $T^{2^{m+1}}=T^{2^m}_1+T^{2^m}_2$ of $M_{m+1}$, we know $M_{m+1}\left( left\left(T^{2^m}_1 \right) \right)=left\left(T^{2^m}_2 \right)$, which determine a couple $\left \langle left\left(T^{2^m}_1 \right), left\left(T^{2^m}_2 \right) \right \rangle$. Let $i=1,T^1_1=\left \{left\left(T^{2^m}_1 \right)\right \}, T^1_2=\left \{left\left(T^{2^m}_2 \right)\right \},R^2=T^1_1+T^1_2$. We perform the following recursive steps:

For $1\le i \le m$, we have $R^{2^i}=T^{2^{i-1}}_1+T^{2^{i-1}}_2, \left \langle T^{2^{i-1}}_1,T^{2^{i-1}}_2 \right \rangle \in D(M_{m+1})$, where $T^{2^{i-1}}_1\subseteq T^{2^m}_1,T^{2^{i-1}}_2\subseteq T^{2^m}_2$. Consider $M_i$, then we have $\left \langle T^{2^{i-1}}_1, M_i\left (T^{2^{i-1}}_1\right)\right \rangle , \left \langle T^{2^{i-1}}_2, M_i\left(T^{2^{i-1}}_2\right)\right \rangle \in D(M_i)$. By Proposition~\ref{prop5.2}, $\left (T^{2^{i-1}}_1+M_i\left( T^{2^{i-1}}_1\right) \right )$ and $\left(T^{2^{i-1}}_1+M_i\left( T^{2^{i-1}}_1\right)\right)$ are both $2^i$-tuples of $M_i$, thus $T^{2^i}_1=\left (T^{2^{i-1}}_1+M_i\left( T^{2^{i-1}}_1\right) \right ) \subseteq T^{2^m}_1$, $T^{2^i}_2=\left(T^{2^{i-1}}_1+M_i\left( T^{2^{i-1}}_1\right)\right)\subseteq T^{2^m}_2$. Therefore, $T^{2^{i-1}}_1,M_i\left( T^{2^{i-1}}_1\right)$, $T^{2^{i-1}}_1,M_i\left( T^{2^{i-1}}_1\right) $ do not overlap. Because $M_{m+1}$ performs on $2$-tuples of $M_i$, we have $\left \langle T^{2^{i-1}}_1, T^{2^{i-1}}_1 \right \rangle, \left \langle M_i\left( T^{2^{i-1}}_1\right),M_i\left( T^{2^{i-1}}_1\right)\right \rangle \in D(M_{m+1}) $, which include $2^i$ couples. Let $R^{2^{i+1}}=T^{2^i}_1+T^{2^i}_2$, then $\left \langle T^{2^i}_1,T^{2^i}_2\right \rangle \in D(M_{m+1})$. Let $i=i+1$ and go to the next step.

After the above steps, we have $2^m$ couples in $D(M_{m+1})$, which is all the couples between $T^{2^m}_1$ and $T^{2^m}_2$.  We do the above steps for each $2^{m+1}$-tuple of $M_{m+1}$, then $M_{m+1}$ is determined, and so is $\left \{ M_1, M_2, \cdots, M_m, M_{m+1}\right \}$.
\end{romanlist}
The induction is completed $\Box$.
\end{proof}

\begin{remark}\label{rem5.3}
We point out that the above lemma only proves the uniqueness, and next we will give a clear example as follows: denote the $i$-th element of $T$ ($0 \le i \le N-1$) as $t_i$, then $\forall 1\le m \le n$, each $2^{m-1}$-bijection of $M_m$  is as $\left \{ t_a, t_{a+1}, \cdots , t_{a+2^{m-1}-1}\right \} \leftrightarrow  \{ t_{a+2^{m-1}}, t_{a+2^{m-1}+1}, \cdots, t_{a+2^m-1} \}$, where $2^m| a$, and $\forall a\le i \le a+2^{m-1}-1$, $M_m(t_i)=t_{i+2^{m-1}}$.
\end{remark}

Then, we can finally prove Proposition~\ref{prop5.3}.

\begin{proof}[Proof of Proposition~\upshape{\ref{prop5.3}}]
In the algorithm, steps 3-5 make $T$ the mapping table for $M_{m}$ which is correct obviously. Steps 6-23 make $T$ $n$-ordered, while steps 10-21 make $j=left\left ( T^{2^{m-1}}_{x2}\right )$. The key is to prove a correspondence that for every two exchanged $2^{m-1}$-tuples $T^{2^{m-1}}_1,T^{2^{m-1}}_2$, the internal mappings are between the corresponding positions, i.e., the ordinal numbers of every two elements exchanged in the two $2^{m-1}$-tuples are the same.
\begin{romanlist}
\item For $m=1$, $T$ is $1$-ordered, then between every two paired $1$-tuples, the correspondence is obvious.
\item After processing $m$, assume that $\forall 1\le l \le m$, the correspondence between every two paired $2^{l-1}$-tuples has been proved. Now process $m+1$. After steps 3-5, for each $2^{m+1}$-tuple $T_x^{2^{m+1}}=T_{x1}^{2^m}+ T_{x2}^{2^m}$ of $M_{m+1}$, $T_{x1}^{2^m}$ and $T_{x2}^{2^m}$ are adjacent without change of the internal arrangement of them thus without change of the correspondence. Now consider steps 10-21 which make $j=left\left ( T^{2^m}_{x2}\right )$, especially steps 15-17. For each time performing steps 15-17, we exchange $T^{\frac{k}{2}}_{y1}$ and $T^{\frac{k}{2}}_{y2}$ as wholes for $1< k\le 2^m$ and each $2^k$-tuple $T^k_y=T^{\frac{k}{2}}_{y1}+T^{\frac{k}{2}}_{y2} \subseteq T_{x2}^{2^m}$, and we will prove that $\forall 1\le l \le m$, it does not change the previous correspondence. For $1\le l < \log_2{k}$, it holds because exchanges between $\frac{k}{2}$-tuples do not change the internal arrangement of $2^{l-1}$-tuples where $2^{l-1}<\frac{k}{2}$; for $l=\log_2{k}$, it holds because after any exchange between $\frac{k}{2}$-tuples without change of their internal arrangement, the correspondence still exists; for $\log_2{k}<l\le m$, it still holds because the internal changes in all $2^{l-1}$-tuple where $2^{l-1} \ge k$ is the same, thus the correspondence remains. Therefore, after step 10-21, we have $j=left\left ( T^{2^m}_{x2}\right )$ and $M_{m+1}(i)=j$, i.e., $T$ is $(m+1)$-ordered. As a necessary step of induction, it must be pointed out that the correspondence for $l=m+1$ holds because when $T$ is $(m+1)$-ordered, $M_{m+1}$ is determined by Lemma~\ref{lemma5.5} and as described by Remark~\ref{rem5.3}.
\end{romanlist}
The induction is completed. As a result, $T$ is $n$-ordered after the whole algorithm $\Box$.
\end{proof}

\subsection{Simplification of semi-orthogonal matrices}\label{subsec5.4}
\noindent
The following definition is the goal of our simplification.

\begin{definition} [\textbf{Ordered type}]
Have an $N=2^n$-order semi-orthogonal matrix $\left | S_N\right |$ and it is a complete semi-cooperative set $G_{\left | S_N\right |}=\left \{I, M_1, M_2,\cdots, M_{N-1} \right \}$, where $M_i$ corresponds to column $i$ of $\left | S_N\right |$. $\left | S_N\right |$ is called an \textbf{ordered type}, if $\forall 1\le i \le N-1, M_i=M^{(i)}$, and of the generator set $\hat{G}_{\left | S_N\right |}=\left \{ M_1, M_2,\cdots, M_{2^{m-1}},\cdots, M_{2^{n-1}} \right \}$, the $n$-ordered mapping table $T=[N]$.
\end{definition}

We will prove the following proposition, which is a solution to Problem~\ref{pro3}.

\begin{proposition}\label{prop5.4}
The following properties of ordered types hold:
\begin{romanlist}
\item  $\forall N = 2^n,n\ge 1$, have an $N$-dimensional column vector $\textbf{\textit{s}}^{(0)}$, then there is only one $N$-order ordered type denoted as $\left | \overline{S}_N\right |$. 
\item Any $N$-order semi-orthogonal matrix can be simplified into $\left | \overline{S}_N\right |$.
\item In $\left | \overline{S}_N\right |$, $\forall 1\le m \le n$, each $2^m\times 2^m$ block in the upper left corner is semi-orthogonal.
\end{romanlist}
\end{proposition}

To prove it, we should first prove the following lemma.

\begin{lemma}\label{lemma5.6}
The following properties of semi-orthogonal matrices $\left | S_N \right |$ hold:
\begin{romanlist}
\item For the complete semi-cooperative set $G_{\left | S_N\right |}=\left \{I, M_1,\cdots, M_{N-1}\right \}$ of an $N$-order semi-orthogonal matrix $\left | S_N\right |$, have an $N$-order row permutation matrix $E$ corresponding to a permutation $\pi:[N]\to[N]$, then the complete semi-cooperative set of $E\left | S_N \right |$ is $G_{E\left | S_N\right |}=\left \{I, EM_1E^T,\cdots, EM_{N-1}E^T\right \}$.
\item For any $N$-order semi-matching operator $M$, if $M(i)=j$, then $\left (EME^T \right )(\pi (i))=\pi (j)$, i.e., if $\left \langle i,j \right \rangle \in D(M)$, then $\left \langle \pi (i), \pi (j)\right \rangle \in D\left( EME^T\right)$, where $E$ is as described above.
\end{romanlist}
\end{lemma}

\begin{proof}
\begin{romanlist}
\item Assume the vector of column $0$ is $\textbf{\textit{s}}^{(0)}$, then $\left | S_N \right | =\left \{\textbf{\textit{s}}^{(0)}, M_1\textbf{\textit{s}}^{(0)},\cdots, M_{N-1}\textbf{\textit{s}}^{(0)}\right \}$, thus 
\noindent
\begin{equation}
\begin{aligned}
\left | S_N \right | &=\left \{E\textbf{\textit{s}}^{(0)}, EM_1\textbf{\textit{s}}^{(0)},\cdots, EM_{N-1}\textbf{\textit{s}}^{(0)}\right \} \\
&=\left \{E\textbf{\textit{s}}^{(0)}, EM_1E^TE\textbf{\textit{s}}^{(0)},\cdots, EM_{N-1}E^TE\textbf{\textit{s}}^{(0)}\right \} \\
&=\left \{{\textbf{\textit{s}}'}^{(0)}, EM_1E^T{\textbf{\textit{s}}'}^{(0)},\cdots, EM_{N-1}E^T{\textbf{\textit{s}}'}^{(0)}\right \}
\end{aligned},
\end{equation}
where ${\textbf{\textit{s}}'}^{(0)}=E\textbf{\textit{s}}^{(0)}$.
\item A row-permutation matrix $E$ is the product of a series of row-exchanging matrices, and $E^T$ corresponds to $\pi^{-1}$. Thus, $\left (EME^T \right )(\pi (i))=\left (\pi M\pi^{-1}\right )(\pi (i))=\pi (M(i))=\pi (j)$ $\Box$.
\end{romanlist}
\end{proof}

Then we can prove Proposition~\ref{prop5.4}.

\begin{proof}[Proof of Proposition~\upshape{\ref{prop5.4}}]
\begin{romanlist}
\item As described in Remark~\ref{rem5.3}, if the mapping table $T=[N]$, i.e., $\forall 0\le i \le N-1, t_i=i$, then $\forall 1\le m \le n$, each $2^{m-1}$-bijection of $M_{2^{m-1}}$  is as $\left \{ a, a+1, \cdots , a+2^{m-1}-1\right \} \leftrightarrow \left \{ a+2^{m-1}, a+2^{m-1}+1, \cdots, a+2^m-1\right \}$, where $2^m| a$, and $\forall a\le i \le a+2^{m-1}-1$, $M_{2^{m-1}}(i)=i+2^{m-1}$. Thus, the generator set $\hat{G}_{\left | \overline{S}_N\right |}=\left \{ M_1, M_2,\cdots, M_{2^{m-1}},\cdots, M_{2^{n-1}} \right \}$ is uniquely determined, and so is the group generated by it. Note that $\forall 0\le i\le N-1,\left \langle 0,i\right \rangle_i\in D(M^{(i)})$, thus all $M_i=M^{(i)}$ are also unique because so is the generator set, then so is $\left | \overline{S}_N\right |$. As a clear example, we show $\left | \overline{S}_N\right |$ as follows:

\noindent
\begin{equation}\label{eq5} 
 \left | \overline{S}_N\right |=\begin{pmatrix}
 a_0 & a_1 & \cdots & a_{2^{n-1}-1} & a_{2^{n-1}} & a_{2^{n-1}+1}  & \cdots & a_{2^n-2} & a_{2^n-1}\\
 a_1 & a_0 & \cdots & a_{2^{n-1}-2} & a_{2^{n-1}+1} & a_{2^{n-1}}  & \cdots & a_{2^n-1} &a_{2^n-2}\\
\vdots & \vdots & \ddots & \cdots & \cdots & \cdots & \cdots & \cdots & \cdots\\
 a_{2^{n-1}-1} & a_{2^{n-1}-2} & \vdots & a_0 & a_{2^n-1} &  a_{2^n-2} &\cdots& a_{2^{n-1}+1}& a_{2^{n-1}} \\
 a_{2^{n-1}} & a_{2^{n-1}+1} & \vdots & a_{2^n-1} & a_0 & a_1& \cdots & a_{2^{n-1}-2} & a_{2^{n-1}-1}\\
 a_{2^{n-1}+1} & a_{2^{n-1}} & \vdots & a_{2^n-2} & a_1 & a_0 & \cdots& a_{2^{n-1}-1} & a_{2^{n-1}-2}\\
 			\vdots & \vdots& \vdots & \vdots& \vdots & \vdots &\ddots & \cdots & \cdots\\
 a_{2^n-2} & a_{2^n-1} & \vdots & a_{2^{n-1}+1}& a_{2^{n-1}-2} & a_{2^{n-1}-1} & \vdots& a_0 & a_1\\
 a_{2^n-1} & a_{2^n-2} & \vdots & a_{2^{n-1}} & a_{2^{n-1}-1} & a_{2^{n-1}-2} & \vdots& a_1 & a_0
\end{pmatrix},
\end{equation} 
where $\textbf{\textit{s}}^{(0)}=\begin{pmatrix}a_0,a_1,\cdots,a_{N-1}\end{pmatrix}^T$.

\item For any $N$-order semi-orthogonal matrix $\left | S_N\right |$, we denote $M_i=M^{(i)}$ and then we get a complete semi-cooperative set $G_{\left | S_N\right |}=\left \{I, M_1,\cdots, M_{N-1}\right \}$. We first perform a column-permutation matrix $F$ on $\left | S_N\right |$ so that the column order of $\left | S_N\right | F$ is the same as $G_{\left | S_N\right |}$. Obviously $G_{\left | S_N\right | F}=G_{\left | S_N\right |}$. Then we consider the generator set $\hat{G}_{\left | S_N\right | F}=\left \{ M_1, M_2,\cdots, M_{2^{m-1}},\cdots, M_{2^{n-1}} \right \}$ of $G_{\left | S_N\right | F}$, and use Algorithm~\ref{algo2} to generate an $n$-ordered mapping table $T=\left\{t_i| i=0,1,\cdots,N-1\right\}$, where $M_m(t_i)=t_j$. Let $\pi:[N]\to [N]$ satisfies $\pi (t_i)=i$, then we get a corresponding row-permutation matrix $E$, and then $\left (EM_mE^T\right)(\pi(t_i))=\pi(t_j)$ by Lemma~\ref{lemma5.6}, i.e., $\left (EM_mE^T\right)(i)=j$. Thus for $E\left | S_N\right | F$, the new mapping table $T'=[N]$, i.e., $E\left | S_N\right | F=\left | \overline{S}_N\right |$.
\item It is easy to see from Eq.~(\ref{eq5}). To prove it, first by Proposition~\ref{prop5.2}, $\forall 1\le m\le n$, in the generated subgroup $\left \langle M_1,M_2,\cdots, M_{2^{m-1}} \right \rangle$, all operators are pairwise $2^m$-similar. In the ordered type, each $2^m$-tuple is as $\left \{ a,a+1,\cdots,a+2^m-1\right \},2^m| a$. Note that $\left \langle 0,i\right\rangle_i\in D(M_i)$ identifies which pairs of $2^m$-tuples are bijectively mapped, i.e., $\forall 1\le i\le 2^m$, $M_i$ corresponds to self-mapping of $2^m$-tuples, thus $\left \{ I,M_1,M_2,\cdots,M_{2^m-2},M_{2^m-1}\right \}=\left \langle M_1,M_2,\cdots, M_{2^{m-1}} \right \rangle$. Therefore, the $2^m\times 2^m$ block $\left| \overline{S}_{2^m} \right|$ in the upper left corner is a semi-orthogonal matrix with a subset $V=\left \{ a_0,a_1,\cdots, a_{2^m-1}\right \} \subseteq U$ $\Box$.
\end{romanlist}
\end{proof}

\subsection{Infeasibility proof of construct a special orthogonal matrix}\label{subsec5.5}
\noindent
The feasibility of the simplification into the ordered type gives the following proposition, which helps to solve Problem~\ref{pro1} thoroughly.

\begin{proposition}\label{prop5.5}
The following two sentences are both necessary conditions for the existence of solutions to Problem~\ref{pro1}:
\begin{romanlist}
\item The $N$-order ordered type $\left |\overline{S}_N\right |$ has a solution.
\item The $\frac{N}{2}$-order ordered type $\left |\overline{S}_{\frac{N}{2}}\right |$ has a solution.
\end{romanlist}
\end{proposition}

\begin{proof}
\begin{romanlist}
\item By Proposition~\ref{prop5.4}, Any $N$-order semi-orthogonal matrix can be simplified into $\left | \overline{S}_N\right |$. For any $N$-order special orthogonal matrice $S_N$, we have $\left | \overline{S}_N\right |=E\left | S_N\right |$ where $E$ is a row permutation matrix, and obviously $\left | \overline{S}_N\right |=\left | ES_N\right |$. Thus if $\left | S_N\right |$ has a solution $S_N$, then $\left | \overline{S}_N\right |$ has a corresponding solution $ES_N$.
\item By Proposition~\ref{prop5.4}, each $2^m\times 2^m$ block in the upper left corner of $\left | \overline{S}_N\right |$ is semi-orthogonal, thus if $\left | S_N\right |$ has a solution, then so do $\left | \overline{S}_N\right |$ and $\left |\overline{S}_{\frac{N}{2}}\right |$ $\Box$.
\end{romanlist}
\end{proof}

In the case of $4$-qubit states, we need a $16$-order special orthogonal matrix, so we have to consider $\left | S_{16} \right |$. We only consider the case as $\textbf{\textit{s}}^{(0)}=\begin{pmatrix} 0, 1,2, \cdots ,15 \end{pmatrix}^T$ because it’s more convenient and cooperative set is more important than the values of elements. We determine the solution existence of $\left | S_{16}\right |$ by Algorithm~\ref{algo1}, and get a result of $FALSE$. On the contrary, the results are both $TRUE$ for $\left | S_4\right |$ and $\left | S_8\right|$. So we have the following theorem finally.

\begin{theorem}\label{theorem1}  $N$-order special orthogonal matrices exist when $N=2,4,8$ and don’t when $N\ge16$.
\end{theorem}

Therefore, it is proved that for $n>3$, it is infeasible to construct a special orthogonal matrix for the DRSP of an arbitrary $n$-qubit state.

There are several similar RSP schemes that use some variant of the above orthogonal matrices. For example, for the RSP of a general state, Xue et al.\cite{2019XueRem} proposed to use matrix of which elements can have complex phases, e.g., 
\noindent
\begin{equation}\textbf{\textit{s}}^{(i)}=
\begin{pmatrix}
\left( -1 \right) ^{\sigma_{0i}}a_{\alpha_{0i}} e^{\imath \varphi_i}\\
\left( -1 \right) ^{\sigma_{1i}}a_{\alpha_{1i}} e^{\imath \varphi_i}\\
\vdots\\
\left( -1 \right) ^{\sigma_{(N-1)i}}a_{\alpha_{(N-1)i}} e^{\imath \varphi_i}
\end{pmatrix},
\end{equation} where $\varphi_i$ are phases. However, at a glance, we will find that the essence of this kind of matrices is similar to that of special orthogonal matrices, and we can prove that the condition for them to exist is equivalent to special orthogonal matrices (see Appendix~\ref{secB} for details), so it’s also infeasible for $n>3$.

\section{Conclusion}\label{sec6}
\noindent
In this paper, we present a polynomial-complexity algorithm to construct a special orthogonal matrix for the DRSP of an arbitrary $n$-qubit state, and prove that for $n > 3$ such matrices do not exist. First, we split the construction problem into two sub-problems, i.e., finding a solution of a semi-orthogonal matrix and generating all semi-orthogonal matrices. Through
giving the definitions and properties of the matching operators, we prove that the orthogonality of a special matrix is equivalent to the cooperation of multiple matching operators, and then the construction problem is reduced to the problem of solving an XOR linear equation system, which reduces the construction complexity from exponential to polynomial level. We prove that each semi-orthogonal matrix can be simplified into a unique ordered type, and then use the proposed algorithm to confirm that the ordered type does not have any solution when $n > 3$. As a corollary, it is infeasible to construct such a special matrix for the DRSP of an arbitrary $n$-qubit state. Considering that DRSP is more valuable than probabilistic RSP, for an arbitrary number $n$ of qubits, it may be more feasible to prepare a certain kind of state, such as an equatorial state\cite{2021ShiCon}, GHZ state\cite{2006DaiCla}, cluster state\cite{2020DuDet}, Brown-type state\cite{2016ChenEco}, etc.

\section*{Declarations}
\noindent
\begin{itemize}
\item \textbf{Conflict of interest} The authors declare that they have no conflict of interest.
\item \textbf{Ethical statement} Articles do not rely on clinical trials. 
\item \textbf{Human and animal participants} All submitted manuscripts containing research which does not involve human participants and/or animal experimentation.
\end{itemize}

\section*{Data availability statement}
Data sharing does not apply to this article as no datasets were generated or analyzed during the current study.

\nonumsection{Acknowledgements}
\noindent
This work is supported by the National Natural Science Foundation of China (62071240), the Priority Academic Program Development of Jiangsu Higher Education Institutions (PAPD), the Innovation Program for Quantum Science and Technology (2021ZD0302902), the High Level Innovation Teams and Excellent Scholars Program in Guangxi Institutions of Higher Education ([2019]52), and the Special Funds for Local Science and Technology Development Guided by the Central Government (ZY20198003).

\nonumsection{References}

\appendix
\noindent\label{secA}

\noindent
In this appendix, the details of Algorithm~\ref{algo1} are shown. This appendix is arranged as follows: we first introduce XOR Gaussian elimination, then present the details of our algorithm, and last we analyze the complexity of the algorithm.

\noindent\textbf{(1) XOR Gaussian elimination}

\noindent
First, we introduce XOR linear equation system as follows, which is solved by using XOR Gaussian elimination:

\begin{definition}[\bf{XOR linear equation system}]
Have an unknown $n$-dimensional vector $x=\begin{pmatrix} x_0, x_1 , \cdots , x_{n-1} \end{pmatrix}^T$, a known $m$-dimensional vector $c=\begin{pmatrix} c_0, c_1 , \cdots , c_{m-1} \end{pmatrix}^T$ and a known $m\times n$ matrix $A=\begin{pmatrix}a_{00} & \cdots & a_{0(n-1)}\\ \vdots & \ddots & \vdots \\a_{(m-1)0} &\cdots& a_{(m-1)(n-1)} \end{pmatrix}$, which called Boolean where $\forall x_i,c_i,a_{ij}\in \{0,1\}$, then we have a \textbf{XOR linear equation system} $Ax\equiv c \pmod{2}$, i.e., 
\noindent
\begin{equation}\label{eqA1}
\begin{pmatrix}a_{00} & \cdots & a_{0(n-1)}\\ \vdots & \ddots & \vdots \\a_{(m-1)0} &\cdots& a_{(m-1)(n-1)} \end{pmatrix} \begin{pmatrix}x_0\\ x_1 \\ \vdots \\ x_{n-1} \end{pmatrix} \equiv \begin{pmatrix}c_0\\ c_1 \\ \vdots \\ c_{m-1} \end{pmatrix} \pmod{2},
\end{equation}
or
\noindent
\begin{equation}\label{eqA2}
\left\{\begin{matrix}
a_{00}x_0 \oplus a_{01}x_1 \oplus \cdots \oplus a_{0(n-1)}x_{n-1}=c_0\\
 \vdots\\
a_{(m-1)0}x_0 \oplus a_{(m-1)1}x_1 \oplus \cdots \oplus a_{(m-1)(n-1)}x_{n-1}=c_{(m-1)}
\end{matrix}\right .,
\end{equation}
where ``$\oplus$" means XOR.
\end{definition}

\begin{definition}[\bf{Row-swap \& Row-XOR})]

To determine the existence of its solutions, we define the following operations for any Boolean matrix:

\begin{romanlist}
\item A \textbf{Row-swap} is that swaps rows $r_i$ and $r_j$, denoted as $r_i \Leftrightarrow r_j$.
\item A \textbf{Row-XOR} is that XOR all elements in row $r_j$ to the corresponding elements in row $r_i$, denoted as $r_i \oplus r_j$.
\end{romanlist}
\end{definition}

\begin{lemma}\label{propA}
If an XOR equation system $A'x\equiv c' \pmod{2}$ is from $Ax\equiv c\pmod{2}$ by a row-swap or row-XOR, then these two systems have the same solutions.
\end{lemma}

\begin{proof}
Obviously for row-swap, so we only need to prove it for row-XOR. Have an XOR equation system below
\noindent
\begin{equation}\label{eqA3}
\left\{\begin{matrix}
a_{00}x_0 \oplus a_{01}x_1 \oplus \cdots \oplus a_{0(n-1)}x_{n-1}=c_0\\
 \vdots\\
 a_{i0}x_0 \oplus a_{i1}x_1 \oplus \cdots \oplus a_{i(n-1)}x_{n-1}=c_i\\
 \vdots\\
 a_{j0}x_0 \oplus a_{j1}x_1 \oplus \cdots \oplus a_{j(n-1)}x_{n-1}=c_j\\
 \vdots\\
a_{(m-1)0}x_0 \oplus a_{(m-1)1}x_1 \oplus \cdots \oplus a_{(m-1)(n-1)}x_{n-1}=c_{(m-1)}
\end{matrix}\right .,
\end{equation}
without loss of generality, assume a row-XOR $r_i \oplus r_j$, then we have
\noindent
\begin{equation}\label{eqA4}
\left\{\begin{matrix}
a_{00}x_0 \oplus a_{01}x_1 \oplus \cdots \oplus a_{0(n-1)}x_{n-1}=c_0\\
 \vdots\\
 (a_{i0}\oplus a_{j0})x_0 \oplus (a_{i1}\oplus a_{j1})x_1 \oplus \cdots \oplus (a_{i(n-1)}\oplus a_{j(n-1)})x_{n-1}=c_i\\
 \vdots\\
 a_{j0}x_0 \oplus a_{j1}x_1 \oplus \cdots \oplus a_{j(n-1)}x_{n-1}=c_j\\
 \vdots\\
a_{(m-1)0}x_0 \oplus a_{(m-1)1}x_1 \oplus \cdots \oplus a_{(m-1)(n-1)}x_{n-1}=c_{(m-1)}
\end{matrix}\right .
\end{equation}
from Eq.~(\ref{eqA3}). Obviously, Eq.~(\ref{eqA3}) and Eq.~(\ref{eqA4}) have the same solutions if 
\noindent
\begin{equation}\label{eqA5}
\left\{\begin{matrix}
 a_{i0}x_0 \oplus a_{i1}x_1 \oplus \cdots \oplus a_{i(n-1)}x_{n-1}=c_i\\
 a_{j0}x_0 \oplus a_{j1}x_1 \oplus \cdots \oplus a_{j(n-1)}x_{n-1}=c_j
\end{matrix}\right .
\end{equation}
and
\noindent
\begin{equation}\label{eqA6}
\left\{\begin{matrix}
 (a_{i0}\oplus a_{j0})x_0 \oplus (a_{i1}\oplus a_{j1})x_1 \oplus \cdots \oplus (a_{i(n-1)}\oplus a_{j(n-1)})x_{n-1}=c_i\\
 a_{j0}x_0 \oplus a_{j1}x_1 \oplus \cdots \oplus a_{j(n-1)}x_{n-1}=c_j
\end{matrix}\right .
\end{equation}
have the same solutions. Let $x=\begin{pmatrix}x_0, x_1 , \cdots , x_{n-1} \end{pmatrix}^T$ be a solution of Eq.~(\ref{eqA3}), i.e., 
\noindent
\begin{equation}\label{eqA7}
\left\{\begin{matrix}
 a_{i0}x_0 + a_{i1}x_1 + \cdots + a_{i(n-1)}x_{n-1}=c_i+2k_i\\
 a_{j0}x_0 + a_{j1}x_1 + \cdots + a_{j(n-1)}x_{n-1}=c_j+2k_j
\end{matrix}\right.
\end{equation} where $k_i,k_j$ are integer. Consequently, we have 
\noindent
\begin{equation}\label{eqA8}
\begin{aligned}
&(a_{i0}\oplus a_{j0})x_0 \oplus (a_{i1}\oplus a_{j1})x_1 \oplus \cdots \oplus (a_{i(n-1)}\oplus a_{j(n-1)})x_{n-1} \\
=&(a_{i0} + a_{j0} +2l_0)x_0 + (a_{i1}+ a_{j1}+2l_1)x_1 + \cdots + (a_{i(n-1)}+a_{j(n-1)+2l_{n-1}})x_{n-1}\\
&+2p\\
= &(a_{i0}x_0 + a_{i1}x_1 + \cdots + a_{i(n-1)}x_{n-1})+(a_{j0}x_0 + a_{j1}x_1 + \cdots + a_{j(n-1)}x_{n-1})\\
&+2(l_0+l_1+\cdots+l_{n-1}+p)\\
=&c_i+2k_i+c_j+2k_j+2(l_0+l_1+\cdots+l_{n-1}+p)\\
=&c_i+c_j+2(Integer)\\
=&c_i \oplus c_j
\end{aligned},
\end{equation}
where $l,p,k$ are all integer, i.e., any solution of Eq.~(\ref{eqA5}) is also a solution of Eq.~(\ref{eqA6}), and vice versa because $a_{jq}\oplus (a_{iq}\oplus a_{jq})=a_{iq}$. Therefore, Eq.~(\ref{eqA5}) and Eq.~(\ref{eqA6}) have the same solutions, so do Eq.~(\ref{eqA3}) and Eq.~(\ref{eqA4}) $\Box$.
\end{proof}

Therefore, the above two operations can be used to solve the equation system. We transform the augmented matrix $\overline{A}=\left ( A | c \right )$ of the equation system into a row simplest matrix, and if $rank(\overline{A}) =rank(A)$, then the solution exists, otherwise not. If the solution exists, then to get a special solution, for each row in eliminated $\overline{A}$, we let all unconstrained variables (i.e., all variables corresponding to non-zero elements except the first non-zero element) be $0$, and let the constrained variable (i.e., the variable corresponding to the first non-zero element) be the constant term, just like the ordinary Gaussian elimination. Similar to the ordinary Gaussian elimination, the time and space complexity of XOR Gaussian elimination are $O(\max^3 (m,n))$ and $O(1)$ respectively.

\noindent\textbf{(2) Algorithm details}

\noindent
In computer, the index of matrix element generally starts from $0$, and $i,j$ generally refer to row and column indexes respectively. We specify that $U=[N]$ for convenience, and if we store any element as a $2$-tuple $(a_k, 0)$ if it is $a_k\in U$ and $(a_k, 1)$ if it is $-a_k\in U$ so that we can handle the case of $0$. For a semi-orthogonal matrix, we handle it as a unfinished special orthogonal matrix of which the elements are all stored as $(a_k, 1)$.

Each couple $\left \langle i,j \right \rangle, i<j$ is stored as a $3$-tuple $\left ( i,  j,\left \langle i,j \right \rangle \right )$ (if $i>j$, we store it as $\left ( j,  i,\left \langle j,i \right \rangle \right )$, where $\left \langle j,i \right \rangle=\left \langle i,j \right \rangle \oplus 1$). We stipulate that if in a couple the logical order and physical order of the two rows are opposite, e.g., $\left \langle i,j \right \rangle$ is stored as $\left (j,i \left \langle j,i \right \rangle \right )$, then set $c\oplus =1$ based on $\left \langle i,j \right \rangle =\left \langle j,i \right \rangle \oplus 1$, where $c$ is the constant term on the right side of a matching equation. The above stipulation is because for any XOR equation $a_0x_0 \oplus a_1x_1 \oplus \cdots \oplus a_{n-1}x_{n-1}=c$ for Boolean variables $x_i $ where $a_i, c$ are all Boolean constants, replace any $x_i$ with $(x_i\oplus 1)$ and then replace $c$ with $c\oplus 1$, then the equation still holds. In addition, we treat every couple in order, so the column index corresponding to each couple $\left \langle x,y \right \rangle_m \in D(M_{m+1})$ is $m\times \frac{N}{2} + num$, where $num$ is the serial number of $\left \langle x,y \right \rangle_m $ in $D(M_{m+1})$.

We have Algorithm~\ref{algoA1} for computing and storing all divisions of matching operators of an $N$-order semi-orthogonal matrix $S$ in table $T$, Algorithm~\ref{algoA2} for inputting matching equations into an $R\times C$ augmented matrix $\overline{A}$ from table $T$, Algorithm~\ref{algoA3} for solving the matching equation system by XOR Gaussian elimination, and Algorithm~\ref{algoA4} for generating a special orthogonal matrix corresponding to $S$ by the solution output from Algorithm~\ref{algoA3}. Based on the above sub-algorithms, we can solve Problem~\ref{pro2} by using Algorithm~\ref{algo1}.

\noindent\textbf{(3) Complexity analysis}

\noindent
In Algorithm~\ref{algoA1}, the hash method is used to speed up, thus its time and space complexity are $O(N^2+N)=O(N^2)$ and $O(N)$ respectively. In Algorithm~\ref{algoA2}, $R=\frac{(N-1)(N-2)N}{8}$ equations are traversed and for each equation, four queries are performed, where each query in the division of a matching operator is of $O(N)$ time complexity, thus the time and space complexity of Algorithm~\ref{algoA2} are $O(N^4)$ and $O(1)$ respectively. For Algorithm~\ref{algoA3}, its time and space complexity are $O(\max^3 (R,C))=O(R^3)=O(\left [\frac{(N-1)(N-2)N}{8}\right ]^3)=O(N^9)$ and $O(1)$ respectively as mentioned before. For Algorithm~\ref{algoA4}, obviously the time and space complexity are $O(N^2)$ and $O(1)$ respectively. Consequently, considering the table $T$, the augmented matrix $\overline{A}$ and the solution vector $X$, the time and space complexity of Algorithm~\ref{algo1} are $O(N^2+N^4+N^9+N^2)=O(N^9)$ and $O(N+1+1+R\times C+(N-1)\times \frac{N}{2}+C-1)=O(N^5)$ respectively.

\begin{algorithm}
\caption{Compute and store all divisions of matching operators of an $N$-order semi-orthogonal matrix $S$. }\label{algoA1}
\begin{algorithmic}[1]
\REQUIRE An $N$-order matrix $S$ where $N=2^n$ and $n >0$, an empty $(N-1)\times \frac{N}{2}$ matrix $T$ of couples $\left (i,j,\left \langle i,j \right \rangle \right ), i<j$, and an empty array $H$ of size $N$.
\ENSURE $S$ is semi-orthogonal and constructed by $U=[N]$.
\FOR{$i \Leftarrow 0$ \textbf{to} $N-1$}
 	\STATE $H[S[i][1]]\Leftarrow i$.
\ENDFOR
\FOR{$j\Leftarrow 1$ \textbf{to} $N-1$}
	\FOR{$i\Leftarrow 0$ \textbf{to} $N-1$}
    	\IF{$\nexists \left \langle i,m \right \rangle\in T[j-1]$}
        	\STATE $t\Leftarrow H[S[i][j]]$.
        	\IF{$S[t][j]=-S[i][0]$}
            	\STATE Store $\left \langle i,t \right \rangle =1$ in $T[j-1]$.
        	\ELSE
         		\STATE Store $\left \langle i,t \right \rangle =0$ in $T[j-1]$.
        	\ENDIF
        \ENDIF
     \ENDFOR
\ENDFOR
\end{algorithmic}
\end{algorithm}

\begin{algorithm}
\caption{Input matching equations into $\overline{A}$ for every two matching operators stored in $T$ table $T$.}\label{algoA2}
\begin{algorithmic}[1]
\REQUIRE An $(N-1)\times \frac{N}{2}$ matrix  $T$ of couples and an $R\times C$ zero Boolean matrix $\overline{A}$ where $R=\frac{(N-1)(N-2)N}{8},C=\frac{N(N-1)}{2}+1$.
\ENSURE  $T$ has been filled by Algorithm~\ref{algoA1}.
\STATE $count\Leftarrow 0$.
\FOR{$0\le a < b <N-1$}
	\FOR{$i\Leftarrow 0$ \textbf{to} $N-1$}
    	\IF{the $4$-tuple involved $i$ has not been traversed}
    		\STATE $c\Leftarrow 1$.
        	\STATE Query $\left \langle i,j \right \rangle$ in $T[a]$.
        	\STATE Query $\left \langle j,k \right \rangle$ in $T[b]$.
        	\STATE Query $\left \langle k,l \right \rangle$ in $T[a]$.
        	\STATE Query $\left \langle l,i \right \rangle$ in $T[b]$.
        	\FOR{$\left \langle x,y \right \rangle_m \in \left \{\left \langle i,j \right \rangle_a,\left \langle j,k \right \rangle_b,\left \langle k,l \right \rangle_a,\left \langle l,i \right \rangle_b\right\}$}
        		\IF{$\left \langle x,y \right \rangle_m$ is stored as $\left (y,x \left \langle y,x \right \rangle \right )$ in $T[m]$} 
            		\STATE $c\Leftarrow c\oplus 1$.
            	\ENDIF
        	\ENDFOR
        	\FOR{$\left \langle x,y \right \rangle_m \in \left \{\left \langle i,j \right \rangle_a,\left \langle j,k \right \rangle_b,\left \langle k,l \right \rangle_a,\left \langle l,i \right \rangle_b\right\}$}
        		\STATE $num \Leftarrow$  the serial number of $\left \langle x,y \right \rangle_m $ in $T[m]$.
            	\STATE $\overline{A}[count][m\times \frac{N}{2} + num]\Leftarrow 1$.
        	\ENDFOR
        	\STATE $\overline{A}[count][C-1]\Leftarrow c$.
        	\STATE $count \Leftarrow count+1$.
		\ENDIF
    \ENDFOR
\ENDFOR
\end{algorithmic}
\end{algorithm}

\begin{algorithm}
\caption{Solve an $R\times C$ augmented matrix $\overline{A}$ by XOR Gaussian elimination.}\label{algoA3}
\begin{algorithmic}[1]
\REQUIRE An $R\times C$ Boolean matrix $\overline{A}$ and an empty vector $X$ of size $C-1$.
\ENSURE  $\overline{A}$ has been filled by Algorithm~\ref{algoA2}.
\STATE $k \Leftarrow 0$.
\FOR{$j\Leftarrow 0$ \textbf{to} $C-1$}
	\FOR{$i \Leftarrow k$ \textbf{to} $R-1$} 
    	\IF{$\overline{A}[i][j] \ne 0$}
        	\STATE $r_i \Leftrightarrow r_k$.
            \FOR{$i\Leftarrow k+1$ \textbf{to} $R-1$ }
           		\IF{$\overline{A}[i][j] \ne 0$}
                    \STATE $r_i \oplus r_k$.
                \ENDIF
            \ENDFOR
            \STATE $k \Leftarrow k+1$.
            \STATE Break.
        \ENDIF
    	\IF{$k=R$}
        	\STATE Break.
        \ENDIF
	\ENDFOR
\ENDFOR
\FOR{ $i \Leftarrow k-1$ \textbf{to} $0$}
    \FOR{ $j \Leftarrow 0$ \textbf{to} $C-1$}
        \IF{$\overline{A}[i][j]\ne 0$}
            \FOR{ $m \Leftarrow i-1$ \textbf{to} $0$}
                \IF{$\overline{A}[m][j] \ne 0$}
                    \STATE $r_m \oplus r_i$.
                \ENDIF
            \ENDFOR
        \ENDIF
    \ENDFOR
\ENDFOR
\IF{ $\exists 0\le j<C-1, \overline{A}[k-1][j]\ne 0$}
    \STATE Store a special solution in $X$.
    \STATE Return $TRUE$.
\ELSE
    \STATE Return $FALSE$.
\ENDIF
\end{algorithmic}
\end{algorithm}

\begin{algorithm}
\caption{Generate a special orthogonal matrix corresponding to a $N$-order semi-orthogonal matrix $S$. }\label{algoA4}
\begin{algorithmic}[1]
\REQUIRE An $N$-order matrix $S$ where $N=2^n$ and $n >0$, an $(N-1)\times \frac{N}{2}$ matrix $T$ of couples, and a vector $X$ of size $C-1$ where $C=\frac{N(N-1)}{2}+1$.
\ENSURE $S$ is semi-orthogonal and constructed by $U=[N]$, $T$ has been filled by Algorithm~\ref{algoA1}, and $X$ is a solution of the matching equation system of $S$ got by Algorithm~\ref{algoA3}.
\STATE Value all couples in $T$ by $X$.
\FOR{$j\Leftarrow 1$ \textbf{to} $N-1$}
    \FOR{$\left \langle x,y \right \rangle \in T[j-1]$}
        \STATE $S[x][j]\Leftarrow  (-1)^{\left \langle x,y \right \rangle \oplus 1}S[y][0]$.
        \STATE $S[y][j]\Leftarrow (-1)^{\left \langle x,y \right \rangle}S[x][0]$.
     \ENDFOR
\ENDFOR
\end{algorithmic}
\end{algorithm}

\appendix
\noindent\label{secB}

\noindent
As is mentioned, a generalization of Problem~\ref{pro1} is to construct an orthogonal matrix of which the elements have complex phases, e.g.,  

\begin{equation}\textbf{\textit{s}}^{(i)}=
\begin{pmatrix}
\left( -1 \right) ^{\sigma_{0i}}a_{\alpha_{0i}} e^{\imath \varphi_i}\\
\left( -1 \right) ^{\sigma_{1i}}a_{\alpha_{1i}} e^{\imath \varphi_i}\\
\vdots\\
\left( -1 \right) ^{\sigma_{(N-1)i}}a_{\alpha_{(N-1)i}} e^{\imath \varphi_i}
\end{pmatrix},
\end{equation} where $\varphi_i$ are phases. We can prove that the condition for it to exist is similar to the special orthogonal matrix.

\begin{lemma}\label{lemmaB}
Given two column vectors 
\noindent
\begin{equation}
\begin{aligned}
\textbf{\textit{s}}^{(i)}=
\begin{pmatrix}
\left( -1 \right) ^{\sigma_{0i}}a_{\alpha_{0i}} e^{\imath \varphi_i}\\
\left( -1 \right) ^{\sigma_{1i}}a_{\alpha_{1i}} e^{\imath \varphi_i}\\
\vdots\\
\left( -1 \right) ^{\sigma_{(N-1)i}}a_{\alpha_{(N-1)i}} e^{\imath \varphi_i}
\end{pmatrix},
\textbf{\textit{s}}^{(j)}=
\begin{pmatrix}
\left( -1 \right) ^{\sigma_{0j}}a_{\alpha_{0j}} e^{\imath\varphi_j}\\
\left( -1 \right) ^{\sigma_{1j}}a_{\alpha_{1j}} e^{\imath \varphi_j}\\
\vdots\\
\left( -1 \right) ^{\sigma_{(N-1)j}}a_{\alpha_{(N-1)j}} e^{\imath \varphi_j}
\end{pmatrix}
\end{aligned},
\end{equation}
then the necessary and sufficient condition for $\textbf{\textit{s}}^{(i)}\perp  \textbf{\textit{s}}^{(j)}$ is that for two new vectors 
\noindent
\begin{equation}
\begin{aligned}
\widetilde{\textbf{\textit{s}}} ^{(i)}=
\begin{pmatrix}
\left( -1 \right) ^{\sigma_{0i}}a_{\alpha_{0i}}\\
\left( -1 \right) ^{\sigma_{1i}}a_{\alpha_{1i}}\\
\vdots\\
\left( -1 \right) ^{\sigma_{(N-1)i}}a_{\alpha_{(N-1)i}} 
\end{pmatrix},
\widetilde{\textbf{\textit{s}}} ^{(j)}=
\begin{pmatrix}
\left( -1 \right) ^{\sigma_{0j}}a_{\alpha_{0j}}\\
\left( -1 \right) ^{\sigma_{1j}}a_{\alpha_{1j}}\\
\vdots\\
\left( -1 \right) ^{\sigma_{(N-1)j}}a_{\alpha_{(N-1)j}} 
\end{pmatrix}
\end{aligned},
\end{equation} we have $\widetilde{\textbf{\textit{s}}} ^{(i)}\perp  \widetilde{\textbf{\textit{s}}} ^{(j)}$.
\end{lemma}

\begin{proof}
\noindent
\begin{equation}
\begin{aligned}
\textbf{\textit{s}}^{(i)\dagger}  \textbf{\textit{s}}^{(j)}&=\sum_{k=0}^{N-1}\left [ (-1)^{\sigma_{ki} }a_{\alpha_{ki}} e^{-\imath\varphi_i }(-1)^{\sigma_{kj} }a_{\alpha_{kj}} e^{\imath\varphi_j}\right] \\
&=\sum_{k=0}^{N-1}\left [ (-1)^{\sigma_{ki}\oplus\sigma_{kj} } a_{\alpha_{ki}} a_{\alpha_{kj}}e^{\imath\left (-\varphi_i + \varphi_j\right ) }  \right] \\
&=e^{\imath\left (-\varphi_i + \varphi_j\right ) } \sum_{k=0}^{N-1}\left [ (-1)^{\sigma_{ki}\oplus\sigma_{kj} } a_{\alpha_{ki}} a_{\alpha_{kj}} \right] \\
&=e^{\imath\left (-\varphi_i + \varphi_j\right ) }  {\left [\widetilde{\textbf{\textit{s}}}^{(i)\dagger}  \widetilde{\textbf{\textit{s}}}^{(j)}\right]}
\end{aligned},
\end{equation}
thus obviously $\Box$.
\end{proof}

By Lemma~\ref{lemmaB}, all the conclusions of the text can be applied here, therefore, the generalized $N$-order special orthogonal matrix as above also exists when $N=2,4,8$ and doesn’t when $N\le16$.

\end{document}